\documentclass[twocolumn%,
]{IEEEtran}
\usepackage[T1]{fontenc}
\usepackage[latin1]{inputenc}
\usepackage[english]{babel} 
\usepackage{amssymb}
\usepackage{amsmath}
\usepackage{amsfonts}
\usepackage{placeins}
\usepackage{pdfsync}
\usepackage{enumerate}
\usepackage{amsthm}
\usepackage{color}
\usepackage{algorithm}
\usepackage{algorithmic}

\makeatletter
\let\NAT@parse\undefined
\makeatother
\usepackage{natbib}

\makeatletter

\usepackage{verbatim}
\usepackage{setspace}
\usepackage{bbm}
\usepackage{amsbsy}

\def \btheta{{\boldsymbol{\theta}}}

\newtheorem{thm}{Theorem}
\newtheorem*{thm*}{} 
\newtheorem{cor}[thm]{Corollary} 
\newtheorem{lem}[thm]{Lemma}
\newtheorem*{lem*}{Lemma}

\newtheorem{dfn}{Definition}

\newtheorem*{rem*}{Remark}

\newcommand{\mathd}{\mathrm{d}}
\newcommand{\mathe}{\mathrm{e}}

\newcommand{\EXP}{\ensuremath{\mathbb{E}}}
\newcommand{\MED}{\ensuremath{\mathrm{median}}}
\newcommand{\IND}{\mathbb{I}}
\newcommand{\KT}{\textsc{kt}}
\newcommand{\Np}{\ensuremath{\mathbb{N}_+}}
\newcommand{\N}{\ensuremath{\mathbb{N}}}
\newcommand{\R}{\ensuremath{\mathbb{R}}}
\newcommand{\PROB}{\ensuremath{\mathbb{P}}}

\newcommand{\M}{\ensuremath{\mathfrak M}}

\def \u{u}

\makeatother

\begin{document}
\title{ 
 About Adaptive Coding on Countable Alphabets: Max-Stable Envelope Classes}

\author{
\IEEEauthorblockN{St\'ephane Boucheron\IEEEauthorrefmark{1}\thanks{\IEEEauthorrefmark{1}supported by Network of Excellence \textsc{pascal ii}, Laboratoire de Probabi\-lit\'es et Mod\`eles Al\'eatoires, Universit\'e Paris-Diderot  \& DMA ENS Ulm, Paris} \and Elisabeth Gassiat\IEEEauthorrefmark{2}\thanks{\IEEEauthorrefmark{2}supported by Network of Excellence \textsc{pascal ii}, Laboratoire de Math\'ematiques d'Orsay, Universit\'e Paris-Sud}
 \and Mesrob I. Ohannessian\IEEEauthorrefmark{3}\thanks{\IEEEauthorrefmark{3}with Microsoft Research - Inria Joint Center. Partially supported by an ERCIM postdoctoral fellowship while at Laboratoire de Math\'ematiques d'Orsay, Universit\'e Paris-Sud.}
}
}

\date{\today}
\maketitle

%%%%%%%%%%%%%%%%%%%%%%%%%%%%%%%%%%%%%%%%%%%%%%%%%%%%%%%%%%%%%%%%%%%%%%%%%%%%%%%%
\begin{abstract}
%%%%%%%%%%%%%%%%%%%%%%%%%%%%%%%%%%%%%%%%%%%%%%%%%%%%%%%%%%%%%%%%%%%%%%%%%%%%%%%%

In this paper, we study the problem of lossless universal source coding for stationary memoryless sources on countably infinite alphabets. This task is generally not achievable without restricting the class of sources over which universality is desired. Building on our prior work, we propose natural families of sources characterized by a common dominating envelope. We particularly emphasize the notion of adaptivity, which is the ability to perform as well as an oracle knowing the envelope, without actually knowing it. This is closely related to the notion of hierarchical universal source coding, but with the important difference that families of envelope classes are not discretely indexed and not necessarily nested.

Our contribution is to extend the classes of envelopes over which adaptive universal source coding is possible, namely by including max-stable (heavy-tailed) envelopes which are excellent models in many applications, such as natural language modeling. We derive a minimax lower bound on the redundancy of any code on such envelope classes, including an oracle that knows the envelope. We then propose a constructive code that does not use knowledge of the envelope. The code is computationally efficient and is structured to use an {E}xpanding {T}hreshold for {A}uto-{C}ensoring, and we therefore dub it the \textsc{ETAC}-code. We prove that the \textsc{ETAC}-code achieves the lower bound on the minimax redundancy within a factor logarithmic in the sequence length, and can be therefore qualified as a near-adaptive code over families of heavy-tailed envelopes. For finite and light-tailed envelopes the penalty is even less, and the same code follows closely previous results that explicitly made the light-tailed assumption. Our technical results are founded on methods from regular variation theory and concentration of measure.

%We describe and analyze a computationally efficient method for encoding sequences from sources over a countable alphabet where source classes are defined by envelope distributions with regularly varying tail functions, that is envelope distributions belong to some maximum domain of attraction. The minimax redundancy over such classes is known to depend on the extreme value index of the envelope distribution, and achieving adaptivity, that is (asymptotic) universality over all those classes simultaneously amounts to circumventing the fact that the extreme value index is unknown. The method we name the \textsc{wac}-code achieves (up to a logarithmic factor) adaptivity by carefully thresholding symbols according to the (intermediate) order statistics in the sample defined by the sequence: symbols smaller than the threshold are encoded using a Krichevsky-Trofimov mixture while larger symbols are encoded using a general purpose code for integers. Though adaptivity is only achieved up to a logarithmic factor, the results presented here go beyond the sharp adaptivity results described in \citep{Bon11,bontemps2012adaptiveB} where a simple thresholding method (the \textsc{ac}-code) consisting of thresholding records is shown to achieve minimax redundancy over all envelope classes with log-concave envelope distribution. The elementary analysis combines arguments from concentration of measure and regular variation theory that have proved effective in the investigation of infinite urn schemes.
\end{abstract}

\bibliographystyle{IEEEtranSN}

%%%%%%%%%%%%%%%%%%%%%%%%%%%%%%%%%%%%%%%%%%%%%%%%%%%%%%%%%%%%%%%%%%%%%%%%%%%%%%%%
%%%%%%%%%%%%%%%%%%%%%%%%%%%%%%%%%%%%%%%%%%%%%%%%%%%%%%%%%%%%%%%%%%%%%%%%%%%%%%%%
%%%%%%%%%%%%%%%%%%%%%%%%%%%%%%%%%%%%%%%%%%%%%%%%%%%%%%%%%%%%%%%%%%%%%%%%%%%%%%%%
\section{Introduction} \label{sec:introduction}
%%%%%%%%%%%%%%%%%%%%%%%%%%%%%%%%%%%%%%%%%%%%%%%%%%%%%%%%%%%%%%%%%%%%%%%%%%%%%%%%
% %%%%%%%%%%%%%%%%%%%%%%%%%%%%%%%%%%%%%%%%%%%%%%%%%%%%%%%%%%%%%%%%%%%%%%%%%%%%%%%%
% %%%%%%%%%%%%%%%%%%%%%%%%%%%%%%%%%%%%%%%%%%%%%%%%%%%%%%%%%%%%%%%%%%%%%%%%%%%%%%%%

The problem we address here is that of \emph{coding} a finite sequence of symbols $x_{1:n}=x_1,...,x_n$, taking values in an (at most) countably infinite \emph{alphabet} $\mathcal{X}$. A \emph{lossless binary source code} (or \emph{code} for short) is a one-to-one map from finite sequences of symbols in $\mathcal X$ of any possible length $n$ to finite sequences of binary $\{0,1\}$ symbols.

We model sequences as being generated by a \emph{source}, defined as a probability measure $\PROB\in\M_1(\mathcal{X}^\N)$ on the set of infinite sequences of symbols from $\mathcal{X}$. We work primarily, for a given $n$, with the finite restriction $\PROB_n$ of this probability measure. That is, $\PROB_n(x_{1:n})$ is the probability of the first $n$ symbols of the random sequence, written $X_{1:n}=X_1,...,X_n$, being equal to $x_{1:n}$. We focus entirely on \emph{stationary memoryless} sources, where $\PROB=\PROB_1^\N$ is a product measure for some $\PROB_1\in \M_1(\mathcal{X})$ called the marginal, itself a probability measure on $\mathcal{X}$. A stationary memoryless source therefore generates independent and identically distributed sequences of symbols. Given a source, the task of source coding is to  minimize the \emph{expected} codelength:
$$
 \EXP[\ell(X_{1:n})] = \sum_{x_{1:n}\in\mathcal{X}^n} \PROB_n(x_{1:n}) \ell(x_{1:n}).
$$
By the source coding theorem, the Shannon entropy of the source
$$
 H(\PROB_n)=-\sum_{x_{1:n}\in {\cal X}^n} \PROB_n(x_{1:n}) \log \PROB_n(x_{1:n})
$$
is a lower bound to the expected codelength of any lossless binary code. (Here and throughout the paper, $\log$ denotes the base-$2$ logarithm). Therefore, one way to measure the performance of any particular code is by its \emph{expected redundancy}, defined as the excess expected length $\EXP[\ell(X_{1:n})]-H(\PROB_n)$. This is meaningful when $\frac{1}{n} H(\PROB_n) = H(\PROB_1)<\infty$, which we assume to be the case throughout.

In this paper, in addition to having to deal with infinite alphabets, we are particularly interested in coding that performs well over a \emph{source class} $\Lambda$, with a common alphabet $\mathcal{X}$, defined as a collection of various probability distributions $\PROB$ on $\mathcal{X}^\N$. We write $\Lambda_n$ to denote the restriction $\{\PROB_n : \PROB \in\Lambda\}$ of $\Lambda$ to distributions on the first $n$ symbols. We now move on to elaborate the classical notion of universality with respect to a source class and then the notion of adaptivity with respect to collections of source classes. We first pass through some more basics about source coding, and we end with a summary of our contributions and an outline the structure of the paper. We use the introduction as a means to introduce all the main notation used in the rest of the paper.

%%%%%%%%%%%%%%%%%%%%%%%%%%%%%%%%%%%%%%%%%%%%%%%%%%%%%%%%%%%%%%%%%%%%%%%%%%%%%%%%
%%%%%%%%%%%%%%%%%%%%%%%%%%%%%%%%%%%%%%%%%%%%%%%%%%%%%%%%%%%%%%%%%%%%%%%%%%%%%%%%
\subsection{Universal source coding}
%%%%%%%%%%%%%%%%%%%%%%%%%%%%%%%%%%%%%%%%%%%%%%%%%%%%%%%%%%%%%%%%%%%%%%%%%%%%%%%%
%%%%%%%%%%%%%%%%%%%%%%%%%%%%%%%%%%%%%%%%%%%%%%%%%%%%%%%%%%%%%%%%%%%%%%%%%%%%%%%%

A code is \emph{uniquely decodable} if any concatenation of codewords can be parsed into codewords in a unique way. The Kraft-McMillan inequality asserts that for a uniquely decodable code over $\mathcal{X}^*$, the codelength map $x_{1:n}\mapsto\ell(x_{1:n})$ satisfies  $\sum_n \sum_{x_{1:n}\in \mathcal{X}^n} 2^{-\ell (x_{1:n})}\leq 1$, and that conversely, given codelengths that satisfy such an inequality, there exists a corresponding uniquely decodable code. The Kraft-McMillan inequality also establishes a deeper correspondence, one between codes over $\mathcal{X}^n$ and probability distributions over $\mathcal{X}^n$: (after normalization) $x_{1:n}\mapsto 2^{-\ell(x_{1:n})}$ defines a probability distribution over $\mathcal{X}^n$, conversely, arithmetic coding \citep*{rissanen:langdon:1984,cover:thomas:1991} allows us to design uniquely decodable codes from any probability distribution over $\mathcal{X}^n$. Therefore we may refer to an arbitrary probability distribution $Q_n\in\M_1(\mathcal{X}^n)$ as a \emph{coding distribution} \citep{cover:thomas:1991}.

%%%%%%%%%%%% TODO 

The correspondence between uniquely decodable codes and probability distributions allows us to describe redundancy as a statistical risk. Indeed, the
expected code length of a coding distribution $Q_n$ is $\EXP_\PROB\left[-\log Q_n(X_{1:n})\right]$, its expected redundancy, when the source is $\PROB$, is the Kullback-Leibler divergence (or relative entropy) between $\PROB_n$ and $Q_n$:
\begin{eqnarray*}
D(\PROB_n, Q_n) &=&\sum_{x_{1:n}\in {\cal X}^n} \PROB_n(x_{1:n}) \log \frac{\PROB_n(x_{1:n})}{Q_n(x_{1:n})} \\
&=&\EXP_{\PROB_n}\left[\log \frac{\PROB_n(X_{1:n})}{Q_n(X_{1:n})}\right] \, .
\end{eqnarray*}
Some authors also call it the \emph{cumulative entropy risk} \citep{MR1604481}.

The theoretically optimal coding probabilities are given by the source $\PROB_n$ itself. And by using methods such as  arithmetic coding, codes corresponding to $\PROB_n$ can be designed to have a redundancy that remains bounded by $1$ for all $n$. 
%But this requires that $\PROB$ is known.

In universal coding,  one attempts to construct a coding distribution $Q_n$ that achieves low redundancy across an entire source class $\Lambda$, without knowing in advance which $\PROB\in\Lambda$ is actually generating the sequence. Such a construction is called coding with respect to $\Lambda$.

To assess a code with respect to a source class, we may adopt one of many perspectives for gauging performance. 
Here,  we study the \emph{maximal redundancy} defined as:
$$
 R^+(Q_n, \Lambda_n)=\sup_{\PROB\in\Lambda} D(\PROB_n, Q_n) \, .
$$
which is essentially as high as the redundancy could grow if $\PROB$ is chosen adversarially at every $n$. Studying this is a way to capture our complete lack of information about which distribution generates the sequence.
 
 The maximal redundancy establishes a uniform rate at which the redundancy grows. The infimum of $R^+(Q_n,\Lambda_n)$ over all $Q_n$, is called the \emph{minimax redundancy} of $\Lambda$:
$$
 R^+(\Lambda_n)=\inf_{Q_n\in {\M}_1\left({\mathcal{X}}^n\right)} R^+(Q_n,\Lambda_n).
$$
The minimax redundancy is a property of the source class $\Lambda$ and represents the best a code could hope for in terms of a guaranteed expected redundancy over the class $\Lambda$.

Universal coding schemes such as the mixture codes developed for memoryless or Markov sources over finite alphabets \citep{krichevsky:trofimov:1981, shtarkov:1987, MR1101099, Rya1984, xie:barron:2000, barron:rissanen:yu:1998, xie:barron:1997, barron:clarke:1994, barron:clarke:1990, Wil98, Gassiat2013} have small and well understood maximal redundancies.
In the simplest setting, that is when considering stationary memoryless sources over a finite alphabet with size $d$, the minimax redundancy scales like $(d-1)/2 \log n $ as the sequence length $n$ tends to infinity. In the language of statistics, classes of sources that can be parametrized by compact subsets of Euclidean spaces are said to be \emph{parametric}. The techniques advocated in the aforementioned references are said to \emph{asymptotically achieve the minimax risk} over the source class in this parametric setting. This is a step beyond strong universality, since the best redundancy decay rate is achieved, and is the notion we strive for in the paper.

%%%%%%%%%%%%%%%%%%%%%%%%%%%%%%%%%%%%%%%%%%%%%%%%%%%%%%%%%%%%%%%%%%%%%%%%%%%%%%%%
%%%%%%%%%%%%%%%%%%%%%%%%%%%%%%%%%%%%%%%%%%%%%%%%%%%%%%%%%%%%%%%%%%%%%%%%%%%%%%%%
\subsection{Adaptive source coding}
%%%%%%%%%%%%%%%%%%%%%%%%%%%%%%%%%%%%%%%%%%%%%%%%%%%%%%%%%%%%%%%%%%%%%%%%%%%%%%%%
%%%%%%%%%%%%%%%%%%%%%%%%%%%%%%%%%%%%%%%%%%%%%%%%%%%%%%%%%%%%%%%%%%%%%%%%%%%%%%%%

Although adaptivity became a major topic in mathematical statistics in the beginning of the early 1990's \citep[see][and references therein]{MR1623559, MR2013911, MR2767163}, the expression adaptive coding barely shows up in articles dedicated to lossless source coding. Source coding research has been mostly concerned with universality. As we have outlined, given a possibly very large collection of sources, a universal code attempts to minimize redundancy, that is the difference between the expected codeword length and the expected codeword length that would be achieved by a code tailored to the source. Adaptive coding considers a more general setting: given a collection of source classes, such that for each class, a good universal coder is available, is it possible to engineer a single coding method that performs well over all classes in the collection?

In the parlance of mathematical statistics, an estimator is said to be \emph{adaptive} over a collection of models or classes if it achieves or at least approaches the minimax risk over all models simultaneously. If we slice the collection of sources of unbounded finite Markov order over a finite alphabet into classes of sources of Markov order $k\in \mathbb{N}$, it is well known that Lempel-Ziv coders are \emph{not} adaptive \citep[see][and references therein]{MR1855254, MR1337757, MR1426235, MR1426236, szpankowski:2001, Gassiat2013}, whereas adaptivity can be achieved over classes of sources of Markov order $k\in \mathbb{N}$ as shown by \citet{Rya1984} who calls adaptivity twice-universality. Such adaptivity is also sometimes called hierarchical universality \citep{MeFe98}, which refers mostly to nested discretely indexed classes. Context-tree weighting is also adaptive \citep{catoni:2004}, and we mention this text individually because it is one of the very few on compression that adopts explicitly the adaptive estimation perspective.
%\citep{FosSti2002} \citep{MR1281931}\citep{gyorfi1993uns}.
%\citep{MR0216548}

There are in fact different flavors of adaptivity in the statistics literature. The textbooks \citep*{MR2767163, MR2013911, MR1623559} define more or less stringent notions of asymptotic adaptivity in the minimax sense. We tune these notions in order to accommodate the context of data compression. Let  $(\Lambda(\mu))$ be a collection of source classes indexed by $\mu\in \mathcal{M}$, where $\mathcal{M}$ is not a necessarily nested or even a discrete set of indices. A sequence $(Q_n)_{n}$ of coding probabilities is said to be \emph{asymptotically adaptive} with respect to a collection $(\Lambda(\mu))_{\mu\in \mathcal{M}}$ of source classes if for all $\mu\in \mathcal{M}$:
\begin{eqnarray} 
\lefteqn{R^+(Q_n, \Lambda_n(\mu)) } \notag\\
&= &\sup_{\PROB \in \Lambda(\mu)} D(\PROB_n,Q_n) \leq (1+o_\mu(1))R^+(\Lambda_n(\mu))\label{eq:adaptive}
\end{eqnarray}
as $n$ tends to infinity. If the inequality \eqref{eq:adaptive} holds with a factor other than $(1+o_\mu(1))$ (that may depend on $\mu$) larger than $1$ to the right, then we say that there is adaptivity \emph{within} this factor. Note that $Q_n$ cannot depend on $\mu$ or else the problem is simply one of universality. \citet*{bontemps2012adaptiveb} describe such an asymptotically adaptive compression scheme for a non-trivial (though restrictive) collection of non-parametric source classes. In order to develop such an adaptive compression technique, it is useful to know the minimax redundancy in each class.

For large collections of massive source classes like the ones we will be handling in this paper, minimax redundancy 
%may   only be known up to a logarithmic factor in the sequence length, but 
may itself grow much faster than any such logarithm. In this case, the logarithmic factor may be meaningfully ignored. In order to accommodate this situation, we present a less stringent criterion of adaptivity. We introduce the following terminology: we call $(Q_n)_{n}$ \emph{asymptotically near-adaptive} with respect to a collection $(\Lambda(\mu))_{\mu\in \mathcal{M}}$ of source classes if for all $\mu\in \mathcal{M}$:
\begin{equation} \label{eq:near-adaptive}
R^+(Q_n, \Lambda_n(\mu))\leq O_{\mu}(\log n ) R^+(\Lambda_n(\mu)) \, ,
\end{equation}
where the constants involved in the $O_{\mu}(\log n )$ term may depend on the source class $\Lambda(\mu)$. Therefore near-adaptivity is adaptivity within a logarithmic factor in the sequence length.

Adaptive source coding raises several challenges: explicit model selection as in \citep*{barron:birge:massart:1999} and source parameter estimation as in two-steps coding schemes \citep*{,rissanen:1984,barron:rissanen:yu:1998} should be avoided so as to make online coding and decoding possible; coding and decoding should be computationally efficient and if possible feasible in linear time. Therefore, in addition to striving to achieve the impressive success of Bayesian coding schemes with respect to parametric classes of sources as demonstrated in the aforementioned papers by Barron \emph{et al.}, we also strive to produce efficient near-adaptive universal codes for large collections of source classes.

%%%%%%%%%%%%%%%%%%%%%%%%%%%%%%%%%%%%%%%%%%%%%%%%%%%%%%%%%%%%%%%%%%%%%%%%%%%%%%%%
%%%%%%%%%%%%%%%%%%%%%%%%%%%%%%%%%%%%%%%%%%%%%%%%%%%%%%%%%%%%%%%%%%%%%%%%%%%%%%%%
\subsection{Contributions and organization of the paper}
%%%%%%%%%%%%%%%%%%%%%%%%%%%%%%%%%%%%%%%%%%%%%%%%%%%%%%%%%%%%%%%%%%%%%%%%%%%%%%%%
%%%%%%%%%%%%%%%%%%%%%%%%%%%%%%%%%%%%%%%%%%%%%%%%%%%%%%%%%%%%%%%%%%%%%%%%%%%%%%%%

Unlike the mostly finite-alphabet results referred to so far, this paper is concerned with adaptive coding over a countably infinite alphabet $\mathcal{X}$ (say the set of positive integers $\Np$ or the set of integers $\N$) as described for example in \cite{MR514346, gyorfi1993uns, FosSti2002, MR2097043, MR2451044, boucheron:garivier:gassiat:2006, garivier:2006, Bon11,Gassiat2013,bontemps2012adaptiveb}. This does not preclude the finite-alphabet case, which becomes a special instance.

When coping with a countably infinite alphabet, even if the source statistics is known, establishing the existence of optimal codes is non-trivial \citep{MR1481062}. More importantly, when we consider universal coding over the class of sources on countably infinite alphabets, even weak universality is not achievable. This was developed in a sequence of papers, starting with early negative results on stationary sources by \cite*{MR514346}, and then also established for memoryless sources by \citeauthor*{gyorfi1993uns} (\citeyear{gyorfi1993uns, MR1281931}). More recently, delicate asymptotic results for coding over large finite alphabets with unknown size have started to appear \citep{MR2097043, SzpWei10,YanBar13}, balancing various finite alphabet sizes and sequence lengths. 

This general difficulty prompted \cite*{boucheron:garivier:gassiat:2006} to first study the redundancy of specific memoryless source classes, namely classes defined by an envelope function. Offline coding techniques for the collection of source classes defined by algebraically vanishing envelopes were introduced in \citep*{boucheron:garivier:gassiat:2006}. \citet{Bon11} designed and analyzed the \textsc{ac}-code (\underline{a}uto-\underline{c}ensuring code). The \textsc{ac}-code has a straightforward structure, it uses a sequence of Krichevsky-Trofimov coders with progressively growing alphabets determined by a threshold that is the maximum of all symbols seen so far: the $i^{\mathrm{th}}$ symbol is either encoded using Krichevsky-Trofimov mixture encoding for alphabet $\{0,\ldots, \max(x_1,\ldots, x_{i-1})\}$, or Elias penultimate encoding if it is the new maximum. \citeauthor{Bon11} proved that this simple code is \emph{adaptive} over the union of classes of sources with exponentially decreasing envelopes. \citet*{bontemps2012adaptiveb} revisited and simplified Bontemps's techniques, and proved moreover that the \textsc{ac}-code is actually adaptive in the sense of \eqref{eq:adaptive} over all classes of sources defined by sub-exponentially decreasing envelopes, that is, envelopes with non-decreasing hazard rate. The \textsc{ac}-code achieves two unexpected benefits: on a practical front it is an online encoding and decoding technique, and on a theoretical front it shows that an effective threshold can be constructed driven by data. The analysis of the \textsc{ac}-code in \citep*{bontemps2012adaptiveb} does not depend on the precise shape of the envelope but strongly benefits from the insights of extreme value theory (EVT) \citep*{FalHusRei11,HaaFei06,MR2108013,Res87} as the minimax redundancy rate of the classes investigated therein asymptotically depends on the slow variation property of the quantile function of the envelope distribution.

A major question that was left open in this work, however, was: is the same adaptivity possible with the much larger class of heavy-tailed envelopes? These envelopes occur often in practice, and are a distinctive property of text and natural language, domains where compression is used extensively. In this paper, we answer this question in the affirmative.

In Section \ref{sec:envelope}, we properly define such heavy-tailed envelope classes. Using the language of EVT, these correspond to \emph{Fr\'echet max-stable} distributions, and are best expressed using the notion of regular variation. In Section \ref{sec:etac-code}, we give the construction of the \textsc{etac}-code, short for \underline{e}xpanding \underline{t}hreshold \underline{a}uto-\underline{c}ensoring code. This is a new computationally efficient code, which builds on the same principle of the \textsc{ac}-code, but uses a new data-driven threshold that expands near the tail of the distribution rather than staying at the maximum. Its thresholding strategy can be summarized in the following way: symbols that are larger than the current threshold tend to be rare for that sequence length and they are encoded using a fixed naive encoder, whereas smaller symbols tend to occur more frequently and they are encoded using the asymptotically maximin Krichevsky-Trofimov encoders tailored to the effective alphabet defined by symbols smaller than the threshold.

In Section \ref{sec:mainresult}, we provide an overview of our main results. The major contribution is the treatment of heavy-tailed envelopes:
\begin{thm*}[\textbf{Theorem \ref{th:main}}]
Over the collection of Fr\'echet max-stable envelope source classes, the \textsc{etac}-code asymptotically achieves the lower bound on the minimax redundancy within a factor logarithmic in the sequence length, and can be therefore qualified as near-adaptive.
\end{thm*}
We also show that for finite and light-tailed envelopes investigated in \citep{bontemps2012adaptiveb}, the same code recovers the adaptivity properties of the \textsc{ac}-code, up to an even slower (roughly $\log\log n$) factor. These results require a lower bound on the minimax redundancy for Fr\'echet max-stable envelope classes, which we give in Section \ref{sec:minimax}, and a detailed analysis of the \textsc{etac}-code, which we perform in Section \ref{sec:etac-analysis}.

The proof techniques combine traditional approaches from information theory \citep{Rya1984,FosSti2002,MR1281931,gyorfi1993uns,Gassiat2013}, regular variation arguments \citep{MR0216548,BinGolTeu89,ohannessian2012rare,ohannessian2012large,BHA13}, as well as concentration inequalities \citep{boluma13}. We collect most of this technical content and proofs within the Appendix.

\section{Envelope Classes} \label{sec:envelope}
%%%%%%%%%%%%%%%%%%%%%%%%%%%%%%%%%%%%%%%%%%%%%%%%%%%%%%%%%%%%%%%%%%%%%%%%%%%%%%%%
%%%%%%%%%%%%%%%%%%%%%%%%%%%%%%%%%%%%%%%%%%%%%%%%%%%%%%%%%%%%%%%%%%%%%%%%%%%%%%%%
%%%%%%%%%%%%%%%%%%%%%%%%%%%%%%%%%%%%%%%%%%%%%%%%%%%%%%%%%%%%%%%%%%%%%%%%%%%%%%%%

We start with the basic definition of an envelope source class.

\begin{dfn}[\textsc{envelope source classes}] \label{dfn:envelope:class}
Let $f$ be a mapping from $\Np$ to $(0,1],$ with $1< \sum_{j\in\Np} f(j)<\infty$. The \emph{envelope class} $\Lambda(f)$ defined by the function $f$ is the collection of stationary memoryless sources with first marginal distribution dominated by $f$:
\begin{eqnarray*}
\Lambda(f)
& =&\Bigl\{ \PROB~:~~\forall j\in \Np,\;\PROB_1\{j\}\leq
f(j)~, \Bigr.
\\
& & \Bigl. \text{ and } \PROB \text{ is stationary and memoryless.}
\Bigr\}\, .
\end{eqnarray*}
\end{dfn}

Envelope classes can be associated with a cumulative distribution, which we call the \emph{envelope distribution}, defined as follows.

\begin{dfn}[\textsc{envelope distribution}]\label{def:env:distribution} Let $f$ be an envelope function.
The associated \emph{envelope distribution} $F$ is defined as
$$
 F(k) = 1 - \sum_{j>k} f(j)
$$
for all $k$ such that $\sum_{j>k} f(j)<1$, and $0$ otherwise. The \emph{tail envelope function} is then defined as the survival function $\overline{F}=1-F$.

Note that the associated probability mass function is equal to $F$ at $\max \{k\colon \sum_{j\geq k} f(j)\geq 1\}$ and does fall below $f$. It coincides with $f$ beyond that point and is zero before it.
\end{dfn}

It is convenient to define a continuous version of the envelope distribution, as follows:
\begin{dfn}[\textsc{smoothed envelope distribution}] \label{dfn:smoothed}
Let $f$ be an envelope function, and let $F$ be its associated envelope distribution. The \emph{smoothed envelope distribution} $F_c$ is a cumulative distribution function on $\R_+$ such that:
 \begin{enumerate}[(i)]
 \item $F_c$ coincides with $F$ on $\mathbb{N}$, and
 \item $F_c$ has a continuous derivative on $\mathbb{R}_{+}$, this derivative is positive at all $x$ such that $F(\lfloor x\rfloor)\in (0,1).$.
\end{enumerate}
Since $F_c$ is effectively an extension of $F$, we allow ourselves to use the $F$ notation to also refer to the smoothed envelope distribution, and mostly avoid the explicit $_c$-subscript notation.
\end{dfn}

For intuition, note that if $Y_c$ is distributed according to the smoothed envelope distribution $F_c$ then $Y=\lceil Y_c\rceil$ is distributed according to the envelope distribution $F$. We do not elaborate on the existence of the smoothed envelope distribution, as explicit constructions may be given by various interpolation methods. We simply remark that point (ii) in Definition \ref{dfn:smoothed} is feasible since envelopes are taken to be strictly positive in Definition \ref{dfn:envelope:class}.

%%%%%%%%%%%%%%%%%%%%%%%%%%%%%%%%%%%%%%%%%%%%%%%%%%%%%%%%%%%%%%%%%%%%%%%%%%%%%%%%
%%%%%%%%%%%%%%%%%%%%%%%%%%%%%%%%%%%%%%%%%%%%%%%%%%%%%%%%%%%%%%%%%%%%%%%%%%%%%%%%
\subsection{Tail Properties and Regular Variation}
%%%%%%%%%%%%%%%%%%%%%%%%%%%%%%%%%%%%%%%%%%%%%%%%%%%%%%%%%%%%%%%%%%%%%%%%%%%%%%%%
%%%%%%%%%%%%%%%%%%%%%%%%%%%%%%%%%%%%%%%%%%%%%%%%%%%%%%%%%%%%%%%%%%%%%%%%%%%%%%%%

In general, we would like to accommodate a large variety of models, yet we do not want models that are too exotic, for both theoretical reasons of tractability and practical reasons of plausibility. With this motivation, we focus on envelope functions that enjoy favorable tail properties. We do this by using the machinery of regular variation and extreme value theory (\textsc{evt} for short). We start with an implicit tail property, but then give a portemanteau theorem that makes explicit various equivalences.

\begin{dfn}[\textsc{maximum domain of attraction}] \label{def:mda}
A (smoothed) distribution function $F$ over $\R$ belongs to a maximum domain of attraction (\textsc{mda}) if there exist sequences $(B_n)_n$ and $(A_n)_n$ with $A_n>0$ and a non-degenerate distribution function $\textsc{gev}$ such that
\begin{math}
 \lim_{n\to \infty} F^n\left( A_n x+B_n \right) = \textsc{gev}(x)
\end{math}
for all $x\in \mathbb{R}$ where $\textsc{gev}$ is continuous, or equivalently if the sequence of distribution functions of $(\max(X_1,\ldots,X_n)-B_n)/A_n$ converges pointwise to $\textsc{gev}$ at every point where $\textsc{gev}$ is continuous.
\end{dfn}

There is in fact much more to belonging to a maximum domain of attraction than this weak (in-law) convergence of rescaled and recentered maxima, and some of this is relevant to adaptive compression as we shall see. Therefore we elaborate more aspects of this property, starting with the fundamental theorem of \textsc{evt} \citep[see][Chapter I]{HaaFei06}. For this, we need to define the following quantities. For all $\gamma \in \mathbb{R},$ let $\textsc{gev}_\gamma(x) = \exp\left( - (1+\gamma x)^{-1/\gamma}\right) $ for $x$ such that $1+\gamma x >0$ (with the convention $\textsc{gev}_0(x)= \exp(-\exp(-x))$. Given a continuous strictly increasing distribution function $F$, let the function $U\colon ]1,\infty) \rightarrow \mathbb{R}$ be a shorthand for the $(1-1/t)$-quantile of $F$, that is:
\begin{equation} \label{eq:quantile}
U(t) =F^{-1}\left(1-{1}/{t}\right) = \overline{F}^{-1}(1/t).
\end{equation}

\begin{thm}[\textsc{fundamental theorem of extreme value theory}] \label{theo:fund:evt}
 Given a distribution function $F$ on $\mathbb{R}$, the following properties are equivalent:
 \begin{enumerate}[(i)]
 \item $F$ belongs to a maximum domain of attraction.

 \item There exist $\gamma\in \mathbb{R}$ and a positive scaling sequence $(A_n)_n$ such that
 \begin{displaymath}
 \lim_{n\to \infty} F^n\left( A_n x+U(n) \right) = \textsc{gev}_\gamma(x)
 \end{displaymath}
 for all $x$ such that $1+\gamma x\geq 0$. This is abbreviated by $F \in \textsc{mda}(\gamma)$.

 \item Conditional excess distributions converge weakly toward a generalized Pareto distribution, that is, there exist $\gamma\in \mathbb{R}$ and a positive scaling function $\sigma$ such that for all $x>0$
 \begin{displaymath}
 \lim_{t\nearrow U(\infty)} \frac{\overline{F}\left( t+ \sigma(t) x \right)}{\overline{F}(t)} = -\log \textsc{gev}_\gamma (x)
 \end{displaymath}

 \item The function $U$ has the \emph{extended regular variation property}, that is, there exists a non-negative measurable function $a$ on $(1,\infty)$
such that for all $x>0$,
 \begin{math}
 \lim_{t\nearrow \infty} \frac{U(tx)-U(t)}{a(t)}
 \end{math}
exists. In that case there exists $\gamma\in \mathbb{R}$ such that the limit is $\int_1^x y^{\gamma-1} \mathrm{d}y.$ This is abbreviated by $U\in \textsc{erv}(\gamma).$
 \end{enumerate}
\end{thm}
Clause (ii) is known as the Fisher-Tippett-Gnedenko Theorem, clause (iii) as the Balkema-de Haan Theorem, and clause (iv) is due to de Haan. The extreme value index $\gamma$ in clauses (ii), (iii), and (iv) is the same. The auxiliary function $a$ in (iv) and $\sigma$ in clause (iii) may be related by choosing $\sigma(t)$ as $a(1/\overline{F}(t)).$

We see therefore that, up to shifting and scaling, the limiting distributions $\textsc{gev}$ of Definition \ref{def:mda} can only be of the form $\textsc{gev}_\gamma$, called \emph{generalized extreme value} distributions. Random variables which have limiting distribution $\textsc{gev}_\gamma$ produce others who do too under the $\max$ operation, which is why such distributions are called \emph{max-stable}. These distributions are known as \emph{Fr\'echet} when $\gamma>0$, \emph{Gumbel} when $\gamma=0$, or \emph{Weibull} when $\gamma<0$.

Note that in this text, we mostly consider envelopes with unbounded support, that is we assume that $U(\infty)=\lim_{t\rightarrow\infty} U(t)$ is infinite. In particular, this means that $\gamma\geq 0$, and we only deal with Fr\'echet and Gumbel limits. In this framework, $U$ has the more basic \emph{regular variation} property: it is regularly varying with index $\gamma$, which we write as $U\in \textsc{rv}(\gamma)$. That is, we have:
$$
\forall x>0,\;\lim_{t\rightarrow +\infty} \frac{U(tx)}{U(t)}=x^{\gamma} \, .
$$
Moreover, if $\gamma>0$, we can choose $\gamma U $ as the auxiliary function $a$ witnessing $U \in \textsc{erv}(\gamma)$ in clause (iv) above. See Appendix \ref{app:regular-variation} for more on regular variation.

%%since $U\in \textsc{erv}(0) \Rightarrow U\in\textsc{rv}(0)$. % \citep[Corollary B.2.13]{HaaFei06}.

%%%%%%%%%%%%%%%%%%%%%%%%%%%%%%%%%%%%%%%%%%%%%%%%%%%%%%%%%%%%%%%%%%%%%%%%%%%%%%%%
%%%%%%%%%%%%%%%%%%%%%%%%%%%%%%%%%%%%%%%%%%%%%%%%%%%%%%%%%%%%%%%%%%%%%%%%%%%%%%%%
\subsection{Max-Stable Envelope Classes}
%%%%%%%%%%%%%%%%%%%%%%%%%%%%%%%%%%%%%%%%%%%%%%%%%%%%%%%%%%%%%%%%%%%%%%%%%%%%%%%%
%%%%%%%%%%%%%%%%%%%%%%%%%%%%%%%%%%%%%%%%%%%%%%%%%%%%%%%%%%%%%%%%%%%%%%%%%%%%%%%%

We are now in position to define the tail properties that we assume for our envelopes. In particular, the smoothed envelope distributions we are interested in belong to some maximum domain of attraction (\textsc{mda}).

\begin{dfn}[\textsc{max-stable envelope classes}] \label{def:max-stable-envelope}
 The envelope class $\Lambda(f)$ with corresponding smoothed envelope distribution function $F$ is said to be a \emph{max-stable envelope class} if $F\in\textsc{mda}(\gamma )$ for some $\gamma \geq 0$. $F$ is said to be a \emph{Fr\'echet} (heavy-tailed) envelope if $\gamma>0$ and to be a \emph{Gumbel} (light-tailed) envelope if $\gamma=0$.
\end{dfn}

\citet{Bon11} and \citet{bontemps2012adaptiveb} considered a strict subset of the set of Gumbel envelopes. In this paper we consider such envelopes more generally, but more fundamentally, we also include the class of Fr\'echet envelopes. Fr\'echet envelopes generalize pure power-law envelopes investigated in
\citep{boucheron:garivier:gassiat:2006}. Indeed, assuming that the class admits a Fr\'echet envelope is equivalent to assuming that the smoothed envelope distribution $F$ is such that
$\overline{F}$ is regularly varying with index $-1/\gamma$ ($\overline{F}\in \textsc{rv}(-1/\gamma)$):
$$
\forall y\leq 1, \quad\lim_{t\rightarrow +\infty} \frac{\overline{F}(ty)}{\overline{F}(t)}=y^{-1/\gamma}\, .
$$
This amounts to there existing a slowly varying function $L$, that is $L\in \textsc{rv}(0)$ (see Appendix \ref{app:regular-variation}), such that $\overline{F}(x)=x^{-1/\gamma} L(x)$.

The max-stability assumption in the definition of this class of sources is instrumental in both the derivation of the minimax redundancy lower bound and the derivation of the upper bound on the redundancy of the \textsc{etac}-code.

%% TODO: Add further justification of why this is an interesting class.

% \input{etac-code}
%%%%%%%%%%%%%%%%%%%%%%%%%%%%%%%%%%%%%%%%%%%%%%%%%%%%%%%%%%%%%%%%%%%%%%%%%%%%%%%%
%%%%%%%%%%%%%%%%%%%%%%%%%%%%%%%%%%%%%%%%%%%%%%%%%%%%%%%%%%%%%%%%%%%%%%%%%%%%%%%%
%%%%%%%%%%%%%%%%%%%%%%%%%%%%%%%%%%%%%%%%%%%%%%%%%%%%%%%%%%%%%%%%%%%%%%%%%%%%%%%%
\section{The ETAC-code} \label{sec:etac-code}
%%%%%%%%%%%%%%%%%%%%%%%%%%%%%%%%%%%%%%%%%%%%%%%%%%%%%%%%%%%%%%%%%%%%%%%%%%%%%%%%
%%%%%%%%%%%%%%%%%%%%%%%%%%%%%%%%%%%%%%%%%%%%%%%%%%%%%%%%%%%%%%%%%%%%%%%%%%%%%%%%
%%%%%%%%%%%%%%%%%%%%%%%%%%%%%%%%%%%%%%%%%%%%%%%%%%%%%%%%%%%%%%%%%%%%%%%%%%%%%%%%

To motivate the construction of the new code, we recall the following theorem from \citep*{boucheron:garivier:gassiat:2006}, which provides an upper-bound on the minimax redundancy of envelope classes and suggests a general design principle for adaptive coding over a collection of envelope classes.
\begin{thm}[\textsc{minimax redundancy upper bound}]
\label{prop:upperbound}
If $\Lambda(f)$ is an envelope class of memoryless sources, with the tail envelope function $\overline{F}$
then:
\begin{displaymath}
R^+(\Lambda_n) \leq \inf_{u : u\leq n} \, \left[ n
 \overline{F}(u)\log e + \frac{u-1}{2}\log n
 \right] + 2\, .
\end{displaymath}
\end{thm}

\citet*{boucheron:garivier:gassiat:2006} also describe settings where this redundancy upper bound is matched by a corresponding lower bound (possibly within a factor of $\log n$). According to Theorem \ref{prop:upperbound}, a threshold $u_n$ should be chosen so as to balance the two terms in the upper bound to have the same growth rate. In particular, the rule of thumb that is evident is to choose $u_n$ such that both of these terms are equal:  $n \overline{F}(u_n)\log e \approx (u_n-1)/2~ \log n $. In the ideal scenario where the envelope distribution is known, this rule of thumb may be combined with known techniques to obtain a code achieving the redundancy upper bound described by Theorem \ref{prop:upperbound}. Namely, these techniques consist of arithmetic coding under the envelope distribution in order to encode symbols larger than the threshold $u_n$, and to encode the sequence of symbols smaller than the threshold $u_n$ using a Krichevsky-Trofimov mixture for alphabet $\{0,1,\ldots, u_n\}$ \citep*[see][for details]{boucheron:garivier:gassiat:2006}.

When the envelope distribution is not known and we strive for adaptivity, one is tempted to replace it with an empirical counterpart. Interestingly, the \textsc{ac}-code, which does operate without knowledge of the envelope, \emph{does not} choose the threshold suggested by Theorem \ref{prop:upperbound} and uses instead the maximum. The \textsc{etac}-code which we propose here does get closer to this principle: by dropping the constants and the $\log n$ term, we choose a threshold $m_{n}$ such that $n\overline{F}(m_n)\approx m_{n}$. Since such a threshold obviously depends on the source class, we construct a threshold from the data to empirically mimic $m_{n}$.

% Assuming the envelope is known does not make universal coding problems trivial since for many envelope distributions such as power law envelopes, the exact minimax redundancy is known up to a logarithmic factor. Nevertheless, the most tantalizing question in that field is concerned with adaptivity issues.

%%%%%%%%%%%%%%%%%%%%%%%%%%%%%%%%%%%%%%%%%%%%%%%%%%%%%%%%%%%%%%%%%%%%%%%%%%%%%%%%
%%%%%%%%%%%%%%%%%%%%%%%%%%%%%%%%%%%%%%%%%%%%%%%%%%%%%%%%%%%%%%%%%%%%%%%%%%%%%%%%
\subsection{Construction of the \textsc{etac}-Code}
%%%%%%%%%%%%%%%%%%%%%%%%%%%%%%%%%%%%%%%%%%%%%%%%%%%%%%%%%%%%%%%%%%%%%%%%%%%%%%%%
%%%%%%%%%%%%%%%%%%%%%%%%%%%%%%%%%%%%%%%%%%%%%%%%%%%%%%%%%%%%%%%%%%%%%%%%%%%%%%%%

The  \textsc{etac}  encoder (Algorithm \ref{etac:encoder}) uses an arithmetic encoder \citep{rissanen:langdon:1984} and a penultimate Elias encoder \citep{MR0373753} as subroutines.
Its input is a message $x_{1:n}$, that is,  a string of positive integers. First $0$  is appended at the end of  the message.
 Then the   \textsc{etac} encoder 
scans the message by iterating over indices $1,\ldots,n+1$ (lines \ref{al:startiter}--\ref{al:enditer}).
% as described in algorithm \ref{etac:encoder}. 

Throughout the iterations, the algorithm maintains a priority queue $PQ$ (a simple binary heap \citep{cormen:leiserson:rivest:stein:2001} is enough).
At iteration corresponding to index $i \in \{1,\ldots,n+1\}$, the priority queue represents the censorship set $\mathcal{C}_i\subset \mathcal{X}$ to be specified, but which only depends on the past symbols from $x_{1:i-1}$, and not on the entire sequence. The current censorship set $\mathcal{C}_i$ consists of symbols not smaller than a threshold $\tau$ . 

\begin{algorithm}[H]
\caption{ETAC encoder}\label{etac:encoder}
%\hrule
\begin{algorithmic}[1]
  \REQUIRE  $x_{1:n}$,   a sequence of positive integers
  \STATE Append $0$ at the end of message $x_{1:n}$
  \STATE \text{Initialize priority queue} $PQ \leftarrow \{ x_1\}$ and threshold $\tau \leftarrow x_1$
  \STATE \text{Initialize counters} (empty dictionary)
  \STATE \text{Initialize the arithmetic encoder using counters}
  \FOR{ $i \in 1,\ldots,\text{length}(x_{1:n}0)$ }\label{al:startiter}
    \STATE $j \leftarrow x_i$
    \IF{$0 < j \leq  \tau$}
      \STATE {feed arithmetic encoder with } $j$
      \STATE emit the output of arithmetic encoder if any
      \STATE $n^j \leftarrow  n^j+1 $   
    \ELSE 
        \STATE {feed arithmetic encoder with} $0$ \\
        \COMMENT{ \emph{this forces the arithmetic encoder to output the whole encoding of the current substring} }
        \STATE emit the  output of the arithmetic encoder
        \STATE {feed  the Elias encoder  with} $\max(1,j-\tau+1)$
        \STATE emit the output of the Elias encoder
        \STATE update or initialize $n^j$
        \STATE \text{update} $PQ$ and $\tau$   
    \ENDIF
  \ENDFOR  \label{al:enditer}
\end{algorithmic}
%\hrule 
\end{algorithm}

If the current  symbol $x_{i} > \tau$, or if $x_{i}=0$, the arithmetic encoder is fed with a $0$ which acts as 
a terminating symbol. This forces the arithmetic encoder to output the total encoding of the portion of the message 
that followed the previously censored symbol. This arithmetic encoding is then emitted by the \textsc{etac} encoder.
 Then $y := \max(x_{i}-\tau+1,1)$ is fed to the Elias penultimate encoder~\citep{MR0373753}. The latter delivers a self-delimited binary encoding of $y$ using $2 \ell(\ell(y))+\ell(y) $ bits where $\ell(z)=\lfloor \log_2(\max(z,1))\rfloor +1$  is the length of the binary encoding of integer $z$.
 Then the \textsc{etac} encoder emits the Elias encoding. If $x_{i}=0$ and only in this case, the input to the Elias encoder is $1$,  
this signals the  end of the message. The queue, the threshold and counters are updated. The censored symbol is inserted in the queue.  If the second smallest symbol in the queue is strictly smaller than the queue size, then the smallest element 
of the queue is popped. Note that there is no need for iterating the popping process as the size of the priority queue is non-decreasing (in the sequel, the random  size of $PQ$ after scanning $n$ symbols is denoted by $M_n$, this random  variable is formally defined by Equation \ref{eq:Mn}, its properties, including monotonicty are discussed afterwards). 

The number of elementary operations required by queue maintenance is proportional to the logarithm of the size of the queue.
As the expected total number of symbols inserted in the queue is sub-linear, the total expected computational cost of the queue maintenance is sub-linear.   

After reading the $i^{th}$ symbol from the message, the alphabet used by the arithmetic encoder is $0, \ldots, \tau$, the state of the arithmetic encoder is a function of the counts $n^0_i =0, n^1_i, \ldots n^{\tau}_i$. Counts may be handled using map or dictionary data structures provided by modern programming languages.

From a bird-eye viewpoint, the scanning process creates a \emph{censored} sequence $\widetilde{x}_{1:n}$ such that every symbol that happens to be in a censorship set is replaced by the special $0$ symbol:
$$
  \widetilde{x}_i = x_i \times \mathbf{1}\{x_i \notin \mathcal{C}_i\}.
$$
The sequence of censored symbols defines a parsing of the message into substrings of
uncensored symbols that are terminated  by $0$. Each substring is encoded by the arithmetic encoder provided with incrementally 
updated sequences of probability vectors (thanks to the counters $n^j, j\leq \max_{k\leq i}(x_k)$). The \textsc{etac} encoder interleaves the outputs of the arithmetic  encoder and the outputs of the penultimate Elias encoder.
Encoding is performed in an incremental way, even though arithmetic coding may require buffering \citep{ShaZaFe06}.  

Let $N$ be the total number of redacted symbols, in the text $i_1,\ldots,i_n$ denote the sequence of indexes  of redacted symbols. Even though  the \textsc{etac} encoder does not produce explicitly these two strings, in the text we will call $C_{\textsc{m}}$ the concatenation of the arithmetic codewords corresponding to the encoding of  $\widetilde{x}_{1:n+1}$, and $C_{\textsc{e}}$ the concatenation of the codewords produced by the Elias encoding of the subsequence $x_{i_{1:N}}$ of redacted symbols.

By construction, the input $y_{1:n}$ of the \textsc{etac} decoder (Algorithm \ref{etac:decoder}) is a binary sequence that can be parsed into a unique sequence of self-delimited codewords originating alternatively from the arithmetic encoder and from the Elias encoder. The functioning of the decoder mirrors the functioning of the encoder. While scanning the input, each time it decodes a symbol, the decoder updates the appropriate counters and maintains a priority queue representing the current   censorship set. 
\begin{algorithm}[H]
\caption{ETAC decoder}
%\hrule
\begin{algorithmic}[1]
\label{etac:decoder}
\REQUIRE $y_{1:n}$ a binary string produced by Algorithm \ref{etac:encoder}
 \STATE $i\leftarrow 1$
 \STATE initialize counters
 \STATE initialize $PQ \leftarrow \varnothing, \tau \leftarrow 1$
 \STATE state $\leftarrow$ arithmetic 
 \WHILE{$i\leq n$}
 \IF{state = Elias}
     \STATE feed the Elias decoder with $y_i$
     \IF{$y_i$ terminates an Elias codeword}
        \STATE emit the output of the Elias decoder
        \STATE update counters, $\tau$ and $PQ$
        \STATE state $\leftarrow$ arithmetic
     \ENDIF 
\ELSE 
     \STATE feed the arithmetic decoder with $y_i$
      \IF{the arithmetic decoder outputs symbols}
        \STATE emit the output of the arithmetic decoder
        \STATE update counters, $\tau$ and $PQ$
      \ENDIF
      \IF{$y_i$ terminates an arithmetic codeword}
        \STATE state $\leftarrow$ Elias
      \ENDIF  
\ENDIF
 \ENDWHILE
 \end{algorithmic}
%\hrule 
\end{algorithm}

In words, the current 
threshold $\tau$ is the smallest symbol in the priority queue.
The queue contains the $M$ largest symbols that have been scanned from the message so far (that is from $x_{1:i}$), and by construction either 
$ \tau = x_{M,i}\leq M $ or $M=i$  and $\tau =  x_{i,i} $.

We now give the details of the censorship set, the encoding of the censored sequence $\widetilde{x}_{1:n}$, and the encoding of the redacted symbols $x_{i_{1:N}}$. Our constructions use the \emph{order statistics} $x_{k,i}$, $k=1,\ldots,i$, defined for various values of $i\in\{1,\cdots,n\}$ as a non-increasing rearrangement of the symbols in the truncated sequence $x_{1:i}$:
$$
\min x_{1:i} = x_{i,i} \leq \cdots \leq x_{1,i} = \max x_{1:i}.
$$

The censorship sets $\mathcal{C}_i$ are constructed as follows. We do not censor for $i=1$, that is we have $\mathcal{C}_1=\varnothing$. Then, for every $i>1$, we censor as follows:
$$
\mathcal{C}_i=\left\{ j\in\N : j > M_{i-1} \right\},
$$
where the (empirical) \emph{threshold} sequence $(M_{i})_{i\in{1:n}}$, defined as
\begin{equation}\label{eq:Mn}
  M_i = \min \left( i, \left\{k\;:\;x_{k,i} \leq k \right\} \right),
\end{equation}
is a sequence of integers such that at each $i \colon 1\leq i\leq n$, $M_i$ can be computed from ${x}_{1:i}$, staying consistent with the past side-information hypothesis.

From the definition of the thresholds $M_i$, note the important fact that both these and the corresponding order statistics $x_{M_i,i}$ are non-decreasing. In fact, as we shall see, $M_i$ and $x_{M_i,i}$ are roughly equivalent and are the empirical version of the thresholds $m_n$ suggested by Theorem \ref{prop:upperbound} and which we subsequently define in Equation \eqref{eq:1}. Therefore we are effectively using as threshold the value $x_{M_{i-1},i-1}$, the $M_{i-1}$-th order statistic at step $i-1$, that is the $M_{i-1}$-th largest symbol among the first $i-1$ symbols. The name ``expanding threshold'' is used to contrast with the \textsc{ac}-code \citep{Bon11} which chooses the $1$-st order statistic (the maximum) as the threshold, whereas here the censure zone ``expands'' to higher order statistics (smaller than the maximum).

The censored sequence $\widetilde{x}_{1:n}$ is encoded into the string $C_{\textsc{m}}$ as follows. We start by appending an extra $0$ at the end of the original censored sequence, to signal the termination of the input. We therefore in fact encode $\widetilde{x}_{1:n}0$ into $C_{\textsc{m}}$. We do this by performing a progressive arithmetic coding \citep{rissanen:langdon:1984} using coding probabilities $Q^{n+1} (\widetilde{x}_{1:n}0)$ given by:
\begin{displaymath}
 \tilde{Q}^{n+1} (\widetilde{x}_{1:n}0) = \tilde{Q}_{n+1}(0\mid x_{1:n}) \prod_{i=0}^{n-1} \tilde{Q}_{i+1} (\widetilde{x}_{i+1}\mid x_{1:i}) \, ,
\end{displaymath}
where the predictive probabilities $\tilde{Q}_{i+1}$ are a variant of Krichevsky-Trofimov mixtures,
\begin{displaymath}
 \tilde{Q}_{i+1} \left(\widetilde{X}_{i+1} = j \mid X_{1:i}=x_{1:i} \right) = \frac{n^j_i+\tfrac{1}{2}}{i+\tfrac{M_{i}+1}{2}}\, .
\end{displaymath}
The $n_i^j$ notation refers to the number of occurrences of symbol $j$ among the first $i$ symbols (in $x_{1:i}$), with the convention that $n^0_i=0$ for all $i$. What these coding probabilities represent, in effect, is a mixture code consisting of progressively enlarging the alphabet based on the thresholds to include symbols $\{0,1,\cdots,M_i\}$, and feeding an arithmetic coder with Krichevsky-Trofimov mixtures over this growing alphabet. Thanks to $M_i$ being determined by the data, the enlargement of the alphabet is performed online and is driven by the order statistics of the symbols seen so far.

The subsequence $x_{i_{1:N}}$ of redacted symbols is encoded into the string $C_{\textsc{e}}$ as follows.  Instead of encoding the symbol values directly, we encode excesses over the thresholds, which are known under the past side-information hypothesis: for each $i\in i_{1:N}$, we encode $x_i-M_{i-1}+1$ using Elias penultimate coding \citep{MR0373753}, where the $+1$ is added to make sure these values are strictly greater than $1$. The extra $0$ initially appended to the message yields the only $1$ that is fed to the arithmetic encoder, it unambiguously signals to the decoder that the $0$ symbol decoded from $C_{\textsc{m}}$ is in fact the termination signal. This ensures that the overall code is instantaneously decodable, and that it therefore corresponds to an implicit coding probability $Q_n$.

\subsection{The Exact Thresholds}
%%%%%%%%%%%%%%%%%%%%%%%%%%%%%%%%%%%%%%%%%%%%%%%%%%%%%%%%%%%%%%%%%%%%%%%%%%%%%%%%
%%%%%%%%%%%%%%%%%%%%%%%%%%%%%%%%%%%%%%%%%%%%%%%%%%%%%%%%%%%%%%%%%%%%%%%%%%%%%%%%

Now that we have constructed the \textsc{etac}-code, let us revisit and compare with Theorem \ref{prop:upperbound}. Recall that the thresholding scheme suggested by this theorem uses as threshold an integer $u_n$ such that $n \overline{F}(u_n)\log e \approx (u_n-1)/2~ \log n$. How well does the \textsc{etac}-code heed this rule of thumb?

Since the encoding is done sequentially, the threshold is adjusted as-we-go, but let us focus primarily on the final threshold $M_n$ as defined by Equation \eqref{eq:Mn} for $i=n$, and the corresponding order statistic $X_{M_n,n}$. From this, it seems that the \textsc{etac}-code uses an apparently suboptimal threshold, by empirically defining $M_n$ as the smallest integer $k$ such that $X_{k,n}\leq k$. Yet, we can make two observations based on this: the first is that $M_n\approx X_{M_n,n}$ and the second is that $\overline{F}_n(X_{M_n,n})\approx X_{M_n,n}/n$, where $\overline{F}_n$ denotes the empirical cumulative distribution function.

We would therefore expect the behavior of both $M_n$ and $X_{M_n,n}$ to closely follow that of their ``exact'' counterpart, that is the value $m_n$ which gives $\overline{F}(m_n)\approx m_n/n$. Therefore the rule of thumb of Theorem \ref{prop:upperbound} is indeed followed up to the constant and logarithmic terms. To analyze the \textsc{etac}-code, it is thus important to study this exact threshold and understand its asymptotic properties. To make this more convenient, instead of using the envelope distribution and working with integer $n$, we can use the \emph{smoothed} envelope distribution $F$ and define for all positive real $t$:
\begin{equation} \label{eq:1}
m(t) = \{ x : \overline{F}(x)= x/t \} = \{ y : U(t/y) = y \}\, 
\end{equation}
and in particular:  $m_n \equiv m(n).$
Note that, by construction, we always have that $m_n$ is non-decreasing, bounded from above by $n$, and satisfies $\overline{F}(m_n)= m_n/n $ and $U(n/m_n) = m_n$.

Using this threshold, Theorem \ref{prop:upperbound} then gives us a minimax redundancy upper bound:
\begin{eqnarray} 
\notag R^+(\Lambda_n) &\leq& m_n \log e + \frac{m_n-1}{2} \log n +2\\
& \leq &   m_n \log n + 2 \qquad (\text{for } n\geq 8)\, \label{eq:minimax-upperbound}
\end{eqnarray}
when $\Lambda$ is a max-stable envelope class with corresponding $U \in \textsc{erv}(\gamma)$ with $\gamma\geq 0$.

Though $m(t)$ is defined in an implicit way, its most relevant properties can be established with little effort thanks to the notion of De Bruijn conjugacy (see Appendix \ref{app:regular-variation}), which plays an important role in the asymptotic inversion of regularly varying functions. The asymptotic behaviors of functions $m$ and $U$ (and therefore $F$) are connected by the following lemma. Namely, $m$ inherits the regular variation property of $U$, and the decomposition of $m$ and $U$ (as products of a slowly varying function and a power function) are related.

\begin{lem}[\textsc{properties of the exact threshold}] \label{lem:implict:func}
Assume that $U : [1,\infty)\longrightarrow \mathbb{R}_+$ is increasing to infinity, is continuously differentiable and that $U \in \textsc{rv}(\gamma)$ with $\gamma\geq 0$ (satisfied under Definition \ref{dfn:smoothed}). Let $m \colon (U(1), \infty) \rightarrow \mathbb{R}_+$ be defined as in Equation \eqref{eq:1}. Then $m$ satisfies:
\begin{enumerate}[(i)]
\item $m$ is well-defined and increasing;
\item $m$ is continuously differentiable;
\item $m(t)\longrightarrow \infty$ and $t/m(t)\longrightarrow \infty$ as $t \longrightarrow \infty$;
\item $m$ is regularly varying with index $\gamma/(\gamma+1)$ $(m\in \textsc{rv}\left(\gamma/(1+\gamma)\right))$. Moreover, if $U(t)=t^\gamma L(t)$ where $L$ is slowly varying, then letting $L_1(t)= L(t^{1/(1+\gamma)})$ and $L_1^*$ be a De Bruijn conjugate of $L_1$,
\begin{displaymath}
 m(t) \sim t^{\gamma/(\gamma+1)} L_1^*(t)^{-1/(1+\gamma)} \text{ as } t \to \infty\, .
\end{displaymath}
\end{enumerate}
\end{lem}
The proof is given in Appendix \ref{app:proofs}.

As for the empirical thresholds $M_n$, or equivalently the corresponding order statistic $X_{M_n,n}$, these are random variables that prove well concentrated around their mean or median values. In particular, if the source is close to the envelope distribution, and if the latter belongs to a max-domain of attraction, the mean value of $M_n$ is close to $m_n$, and is a regularly varying function that reflects the tail behavior of the envelope distribution. These results are presented in Appendix \ref{app:threshold-properties}, and used in the analysis of the \textsc{etac}-code in Section \ref{sec:etac-analysis}.

% \input{main-results}
%%%%%%%%%%%%%%%%%%%%%%%%%%%%%%%%%%%%%%%%%%%%%%%%%%%%%%%%%%%%%%%%%%%%%%%%%%%%%%%%
%%%%%%%%%%%%%%%%%%%%%%%%%%%%%%%%%%%%%%%%%%%%%%%%%%%%%%%%%%%%%%%%%%%%%%%%%%%%%%%%
%%%%%%%%%%%%%%%%%%%%%%%%%%%%%%%%%%%%%%%%%%%%%%%%%%%%%%%%%%%%%%%%%%%%%%%%%%%%%%%%
\section{Main Results} \label{sec:mainresult}
%%%%%%%%%%%%%%%%%%%%%%%%%%%%%%%%%%%%%%%%%%%%%%%%%%%%%%%%%%%%%%%%%%%%%%%%%%%%%%%%
%%%%%%%%%%%%%%%%%%%%%%%%%%%%%%%%%%%%%%%%%%%%%%%%%%%%%%%%%%%%%%%%%%%%%%%%%%%%%%%%
%%%%%%%%%%%%%%%%%%%%%%%%%%%%%%%%%%%%%%%%%%%%%%%%%%%%%%%%%%%%%%%%%%%%%%%%%%%%%%%%

We now give the main contributing result of the paper, which is the near-adaptivity of the \textsc{etac}-code on the collection of Fr\'echet (heavy-tailed) max-stable envelope source classes. The components of this result are then presented in detail in the rest of the paper, in terms of the minimax redundancy lower bound in Section \ref{sec:minimax} and the analysis of the \textsc{etac}-code in Section \ref{sec:etac-analysis}. We also give a somewhat stronger adaptivity of the \textsc{etac}-code, when restricted to a sub-class of Gumbel (light-tailed) max-stable envelopes. In particular, the overhead is not logarithmic in the sequence length, but logarithmic in the minimax redundancy, which grows much slower (roughly $\log\log n$). Within this setting, this shows that there is no major loss in switching to the new code, which does not explicitly make the light-tailed assumption, from the \textsc{ac}-code of \citep{bontemps2012adaptiveb}, which does.

%%%%%%%%%%%%%%%%%%%%%%%%%%%%%%%%%%%%%%%%%%%%%%%%%%%%%%%%%%%%%%%%%%%%%%%%%%%%%%%%
%%%%%%%%%%%%%%%%%%%%%%%%%%%%%%%%%%%%%%%%%%%%%%%%%%%%%%%%%%%%%%%%%%%%%%%%%%%%%%%%
\subsection{Near-adaptivity to Fr\'echet max-stable envelope source classes}
%%%%%%%%%%%%%%%%%%%%%%%%%%%%%%%%%%%%%%%%%%%%%%%%%%%%%%%%%%%%%%%%%%%%%%%%%%%%%%%%
%%%%%%%%%%%%%%%%%%%%%%%%%%%%%%%%%%%%%%%%%%%%%%%%%%%%%%%%%%%%%%%%%%%%%%%%%%%%%%%%

Our main result can be stated as:
\begin{thm}[\textsc{frechet near-adaptivity of the etac code}] \label{th:main}
Let $Q_n$ denote the coding probability defined by the \textsc{etac}-code, let $\Lambda$ be a Fr\'echet max-stable envelope class with ultimately non-increasing envelope and with corresponding exact threshold sequence $(m_n)_{n\in\N_+}$. We then have that there exists a constant $\kappa_\Lambda$ (that may depend on $\Lambda$) such that:
\begin{eqnarray*}
  (\kappa_\Lambda+o_\Lambda(1)) {m_n} &\leq& R^+(\Lambda_n) \leq  R^+(Q_n,\Lambda_n) \\
&\leq &(5/2+o_{\Lambda}(1)) m_n \log n.
\end{eqnarray*}
In particular, the \textsc{etac}-code is asymptotically near-adaptive (cf. Equation \eqref{eq:near-adaptive}):
$$
  R^+(Q_n,\Lambda_n) \leq (5/2\kappa_\Lambda +o_\Lambda(1))\log n ~ R^+(\Lambda_n).
$$
\end{thm}

We provide here a guideline proof of this theorem, which relies on components presented in Sections \ref{sec:minimax} and \ref{sec:etac-analysis}, as well as their details in the appendices.
\begin{proof}
The lower bound on the minimax redundancy is given by Theorem \ref{thm:minimax} of Section \ref{sec:minimax}. To sketch its proof, note that we first use a maximin Bayes redundancy approach to prove that the minimax redundancy is lower bounded by the number $K_{n}$ of distinct symbols that appear in the sequence $X_{1:n}$. We then show that for sources belonging to max-stable envelope classes, $K_n$ and $M_n$ (and therefore $m_n$, by Appendix \ref{app:threshold-properties}) are asymptotically within constant factors of each other.

The upper bound on the redundancy of the \textsc{etac}-code follows from the results of Section \ref{sec:etac-analysis}, where we separately analyze the codelengths of the arithmetically mixture-encoded censored sequence $C_{\textsc{m}}$ and the Elias-encoded individual redacted symbols $C_{\textsc{e}}$. We show here how to combine this analysis to provide the upper bound.

Assume that we are dealing with a particular source from a max-stable envelope class $\Lambda$. For the mixture-encoded censored sequence, Equation \eqref{eq:mixture-code} bounds the difference between the expected length of $C_{\textsc{m}}$ and the optimal codelength given by the Shannon entropy:
\begin{eqnarray*}
  \EXP \left[ \ell (C_{\textsc{m}}) +\log \PROB_n(X_{1:n})\right]
%      &\leq & \frac{\EXP X_{M_n,n}+1}{2} \log n + 2 \\
      &\leq & (1/2+o_\Lambda(1))~m_n \log n \, .
\end{eqnarray*}

Meanwhile, Lemma \ref{lem:elias-code} gives a final breakdown of the expected length of $C_{\textsc{e}}$. In particular, since we are considering only Fr\'echet distributions, we can use Equation \eqref{eq:elias-bound-frechet}:
\begin{eqnarray*}
  \EXP [\ell(C_{\textsc{e}})]
%     &\leq& (2+o_{\Lambda}(1)) \sum_{i=1}^n \frac{m_i \log m_i }{i} \\
    &\leq& (2+o_{\Lambda}(1)) ~ m_n \log n \, .
\end{eqnarray*}

Merging the two upper bounds leads to:
$$
   \EXP \left[ \ell(C_{\textsc{e}})+ \ell (C_{\textsc{m}}) +\log \PROB_n(X_{1:n}) \right]\leq (5/2+o_{\Lambda}(1)) ~ m_n\log n \, ,
$$
and since this bound does not depend on the particular source but only on the source class, we can take a supremum of the left-hand side over the entire class to obtain an upper bound on $R^+(Q_n,\Lambda_n)$.
\end{proof}

As a straightforward consequence of Theorem \ref{th:main} and point (iv) of Lemma \ref{lem:implict:func}, we have exact rates of growth of the lower bound on the minimax redundancy and the redundancy of the \textsc{etac}-code, in terms of the regular variation properties of the envelope:
\begin{cor}
Let $Q_n$ denote the coding probability defined by the \textsc{etac}-code, let $\Lambda$ be a Fr\'echet max-stable envelope class with ultimately non-increasing envelope
and  with smoothed envelope distribution in the maximum domain of attraction
$\textsc{mda}(\gamma )$ for some $\gamma > 0$. Then there exist a slowly varying function $L_\Lambda$ and a constant $\kappa_\Lambda$ (both depending on the envelope that defines $\Lambda$), such that:
\begin{eqnarray*}
\kappa_\Lambda ~ L_\Lambda(n) ~ n^{\gamma/(\gamma +1)} &\leq & R^+(\Lambda_n) \\
&\leq & R^+(Q_n,\Lambda_n) \\
&\leq & (5/2+o_{\Lambda}(1))\log n~ L_\Lambda(n) ~ n^{\gamma/(\gamma +1)} .
\end{eqnarray*}
\end{cor}
This corollary is particularly informative, since it shows that both the minimax redundancy and the \textsc{etac}-code redundancy grow as powers of $n$, and therefore the logarithmic factor in the definition of near-adaptivity is not an unreasonable relaxation to the notion of adaptivity in the context of Fr\'echet max-stable envelope classes. Note also the vanishing per-symbol redundancy, at a rate of roughly $n^{-1/(1+\gamma)}$ which is slower the heavier the tail (the larger $\gamma$) is.

%%%%%%%%%%%%%%%%%%%%%%%%%%%%%%%%%%%%%%%%%%%%%%%%%%%%%%%%%%%%%%%%%%%%%%%%%%%%%%%%
%%%%%%%%%%%%%%%%%%%%%%%%%%%%%%%%%%%%%%%%%%%%%%%%%%%%%%%%%%%%%%%%%%%%%%%%%%%%%%%%
\subsection{Near-adaptivity to light-tailed envelope source classes}
%%%%%%%%%%%%%%%%%%%%%%%%%%%%%%%%%%%%%%%%%%%%%%%%%%%%%%%%%%%%%%%%%%%%%%%%%%%%%%%%
%%%%%%%%%%%%%%%%%%%%%%%%%%%%%%%%%%%%%%%%%%%%%%%%%%%%%%%%%%%%%%%%%%%%%%%%%%%%%%%%

In this section, we tie the results of this paper with those of \citet{bontemps2012adaptiveb}, where an explicit light-tailed assumption was made. This is the notion of an envelope distribution $F$ that has \emph{non-decreasing hazard rate}, that is it can be associated with a log-convex smoothed tail function. The terminology comes from the notion of \emph{hazard function} whose derivative being non-decreasing is equivalent to this log-convexity condition. It is worth noting that this means that $F$ itself is, almost, log-concave. These distributions are a rich subset of light-tailed distributions. In particular, geometric envelopes are at the boundary of such distributions, as they exhibit log-linear smoothed tail functions.

The contribution of the \textsc{ac}-code presented in \citep{bontemps2012adaptiveb} is that this code is \emph{adaptive}, in the sense of Equation \eqref{eq:adaptive}, to the collection of classes with non-decreasing hazard rate. The performance of the \textsc{ac}-code on such classes may be understood in a very intuitive way: the \textsc{ac}-code encodes the $n$-th symbol in a way that is not more expensive than encoding a symbol from a source on an alphabet of size $U(n)= F^{-1}(1-1/n)$, that is with redundancy $U(n)/(2n)$ bits. The \textsc{ac}-code can perform in this way for two reasons: with overwhelming probability, the largest sample in a sequence of length $n$, is not larger than $U(n)$; on many sources in such a class, with high probability, most of the symbols that are smaller than $U(n)$ do occur in a sequence of length $n$, there is no penalty in coding as if the actual alphabet were of size $U(n)$.

The \textsc{etac}-code does not take such a simplistic approach, it attempts to calibrate the effective alphabet size in a much more cautious way. An intuitive interpretation of the empirical threshold $M_n$ is the following: symbols larger than $M_n$ have low empirical frequency in the sequence, they may be encoded with the general purpose code; symbols smaller than $M_n$ tend to have larger empirical frequency, and on some sources from the envelope classes considered in this paper, a large proportion of the symbols that are smaller than $m_n$ do occur in a typical sequence (this observation is documented in the literature \citep*{ArKnPr06,MR2582705,BHA13}). Up to the Elias encoding, the \textsc{etac}-code encodes a sequence of length $n$ as if the actual alphabet were of cardinality $m_n$. The choice of $m_n$ balances the cost of escaping large symbols and the overhead incurred by oversizing the effective alphabet.

On the other hand, \citet{bontemps2012adaptiveb} establish that for non-decreasing hazard rate envelope classes, $U(t)= F^{-1}(1-1/t)$ is not only slowly varying but also enjoys the special property that, according to \citet{bojanic1971slowly}, the De Bruijn conjugate $U^*$ of $U$ is asymptotically equivalent to $1/U$. By Lemma \ref{lem:implict:func}, this in turn implies that $\lim_{t\to \infty}m(t)/U(t)= 1.$ Operationally, this means that choosing the threshold as $M_n \approx X_{M_n,n}$ (\textsc{etac}-code) or as $X_{1,n}$ (\textsc{ac}-code) asymptotically does not make a difference as far as coding envelope classes defined by such light-tailed envelopes. This entails  \citep[see][]{bontemps2012adaptiveb} the fact that the minimax redundancy of such classes is asymptotically not smaller than $\log (e) \int_1^n U(x)/(2x) \mathrm{d}x\geq U(n)\log (n)/4$.

Therefore, we expect the \textsc{etac}-code to perform well, despite its cautious approach. The following theorem establishes precisely that: up to a $\log m_n \approx \log\log n$ factor, the \textsc{etac}-code is asymptotically adaptive with respect to envelope classes defined by envelope distributions with non-decreasing hazard rate.

\begin{thm}[\textsc{non-decreasing hazard rate near-adaptivity of the etac code}] \label{thm:wac:hazrd}
Let $Q_n$ denote the coding probability defined by the \textsc{etac}-code, let $\Lambda$ be an envelope class such that the envelope has the non-decreasing hazard rate property, with corresponding exact threshold sequence $(m_n)_{n\in \N_+}$. We then have that:
\begin{eqnarray*}
  (1/4+o_\Lambda(1)) ~ m_n\log n &\leq & R^+(\Lambda_n) \\
& \leq &  R^+(Q_n,\Lambda_n) \\
&\leq &(2+o_{\Lambda}(1)) ~ m_n\log n\log m_n \, .
\end{eqnarray*}
In particular, the \textsc{etac}-code is not only asymptotically near-adaptive (cf. Equation \eqref{eq:adaptive}, noting that $m_n\leq n$):
$$
  R^+(Q_n,\Lambda_n) \leq (8+o_\Lambda(1))\log m_n R^+(\Lambda_n),
$$
but furthermore the multiplicative factor is of order $O_\Lambda(\log\log n)$.
\end{thm}
\begin{proof}
The proof of the minimax redundancy lower bound is given in \citep{bontemps2012adaptiveb}. As for the redundancy upper bound, the only difference with the proof of Theorem \ref{th:main} is to use the weaker Elias codelength bound given by Equation \eqref{eq:elias-bound-general}:
\begin{eqnarray*}
  \EXP [\ell(C_{\textsc{e}})]
%     &\leq& (2+o_{\Lambda}(1)) \sum_{i=1}^n \frac{m_i \log m_i }{i} \\
    &\leq& (2+o_{\Lambda}(1)) ~ m_n \log n \log m_n\, .
\end{eqnarray*}

The near-adaptivity follows immediately. As for the claim that $m_n = O_\Lambda(\log n)$, it follows from the fact that $m(t)/U(t)\to 1$ as shown in \citep{bontemps2012adaptiveb}, recalling that $U(t)=\overline{F}^{-1}(t)$, where $\overline{F}$ has a sub-exponential tail.
\end{proof}
% \input{minimax-lower-bound}
%%%%%%%%%%%%%%%%%%%%%%%%%%%%%%  Discussion

Theorems \ref{th:main} and \ref{thm:wac:hazrd} raise several questions. Between heavy-tailed envelope functions handled by Theorems \ref{th:main}
and very light tailed envelope functions handled by Theorem \ref{thm:wac:hazrd} , lie an intermediate family of envelope functions with slowly varying tail quantile functions ($U(tx)/U(t)\to 1$ as $t \to \infty)$ for all $x>0$) but with decreasing hazard rate. If we consider sampling from the associated envelope distribution, the literature dedicated to infinite urn schemes \citep[See][]{MR0216548,GneHanPit07,BHA13} shows that as $n$ tends to infinity, the number of rare symbols -- that are likely to be censored and to enter the priority queue maintained by the \textsc{etac} encoder -- is not stochastically bounded, but it tends to be negligible with respect to the number of distinct symbols in the sample. The \textsc{ac}-code is not likely to be adaptive with respect to envelope classes defined by such envelope distributions. The minimax redundancy of such envelope classes remains to be determined, and so is the performance of the \textsc{etac} code over thoses classes. Indeed, a very natural question raised by the advances reported in the present paper, is the cost of adaptivity in compression against countable alphabets. In density estimation \citep{lepski1992}, or tail index estimation \citep{CarKim14,BouTho15} for example, there are problems where adaptive estimation  suffers a logarithmic loss with respect to minimax risk. We still do not know whether this is the case for adaptive compression against envelope classes.

%%%%%%%%%%%%%%%%%%%%%%%%%%%%%%%%%%%%%%%%%%%%%%%%%%%%%%%%%%%%%%%%%%%%%%%%%%%%%%%%
%%%%%%%%%%%%%%%%%%%%%%%%%%%%%%%%%%%%%%%%%%%%%%%%%%%%%%%%%%%%%%%%%%%%%%%%%%%%%%%%
%%%%%%%%%%%%%%%%%%%%%%%%%%%%%%%%%%%%%%%%%%%%%%%%%%%%%%%%%%%%%%%%%%%%%%%%%%%%%%%%
\section{Minimax Redundancy of Fr\'echet Envelope Classes} \label{sec:minimax}
%%%%%%%%%%%%%%%%%%%%%%%%%%%%%%%%%%%%%%%%%%%%%%%%%%%%%%%%%%%%%%%%%%%%%%%%%%%%%%%%
%%%%%%%%%%%%%%%%%%%%%%%%%%%%%%%%%%%%%%%%%%%%%%%%%%%%%%%%%%%%%%%%%%%%%%%%%%%%%%%%
%%%%%%%%%%%%%%%%%%%%%%%%%%%%%%%%%%%%%%%%%%%%%%%%%%%%%%%%%%%%%%%%%%%%%%%%%%%%%%%%

We now lower bound the minimax redundancy with respect to the envelope class $\Lambda(f)$ when the envelope function $F$ is Fr\'echet. In this section, we make the additional assumption that $f$ is ultimately monotonically non-increasing. This is primarily to make the presentation more transparent when relating the regular variation properties of various functions, namely $f$, $F$, and the distribution $G$ that ensues from the Bayes construction.

We use the standard approach of the relationship between minimax and maximin redundancies. In particular, consider a set $\mathcal{P}=\{P_{\btheta}, \btheta\in \Theta\}$ of memoryless sources over the countable alphabet $\Np$ indexed by a parameter space $\Theta$ and let $\pi$ be a (prior) probability measure on $\Theta$. We call $(\mathcal P,\pi)$ a \emph{Bayes model}. If the parameter $\btheta$ is chosen according to $\pi$ and subsequently a sequence $X_{1:n}$ of length $n$ is observed from the source $P_{\btheta}$, then the \emph{Bayes redundancy} is the mutual information between $\btheta$ to $X_{1:n}$. Of fundamental importance is the fact (see, for example, \cite{barron:clarke:1990}) that the minimax redundancy is lower bounded by the Bayes redundancy with respect to any choice of prior probability distribution:
$$
 I(\btheta,X_{1:n}) \leq R^+(\mathcal{P}^n) \, .
$$
Moreover, whenever $\mathcal{P} \subset \Lambda$, we have $R^+(\mathcal{P}^n) \leq R^+(\Lambda_n)$, and we can engineer a lower bound to the minimax redundancy by properly choosing the Bayes model $(\mathcal{P},\pi)$. In the remainder of this section we start by doing precisely that, we then bound the resulting mutual information by the expected number of distinct symbols in the sequence, and lastly we relate the growth of the latter to the index of regular variation to establish a lower bound that matches the redundancy of the \textsc{ETAC}-code up to a logarithmic factor.

%%%%%%%%%%%%%%%%%%%%%%%%%%%%%%%%%%%%%%%%%%%%%%%%%%%%%%%%%%%%%%%%%%%%%%%%%%%%%%%%
%%%%%%%%%%%%%%%%%%%%%%%%%%%%%%%%%%%%%%%%%%%%%%%%%%%%%%%%%%%%%%%%%%%%%%%%%%%%%%%%
\subsection{Building a Bayes model}
%%%%%%%%%%%%%%%%%%%%%%%%%%%%%%%%%%%%%%%%%%%%%%%%%%%%%%%%%%%%%%%%%%%%%%%%%%%%%%%%
%%%%%%%%%%%%%%%%%%%%%%%%%%%%%%%%%%%%%%%%%%%%%%%%%%%%%%%%%%%%%%%%%%%%%%%%%%%%%%%%

In an appropriate Bayes model, we would like each $P_\btheta$ to be a member of $\Lambda(f)$ in an intuitively `worst-case' fashion: we want to capture the tail behavior dictated by $f$. The parameters can then simply `dither' around this tail.

Let $\Theta = \{0,1\}^\N$ be the space of all $0$-$1$ sequences. For any such sequence $\btheta=(\theta_k)_{k\in\N}$ define $P_\btheta \in \mathcal P$ as, for each $j\in \Np$:
\begin{eqnarray*} \label{eq:bayes-model}
\lefteqn{P_\btheta(j) } \\ &= &\left\{
\begin{array}{lcl}
 f(j)/Z &\quad& \mathrm{for\ every}\ j<j_0\\
 \min\limits_{t\in\{0,1\}} f(j_0+2k+t) &\quad& \mathrm{when}\ j=j_0+2k+\theta_k\\
&&\mathrm{for\ some}\ k\in\N\\
 0 &\quad& \mathrm{when}\ j=j_0+2k+(1-\theta_k)\\
&&\mathrm{for\ some}\ k\in\N\\
\end{array}
\right.
\end{eqnarray*}
where
$$
 Z = \frac{\sum_{j<j_0} f(j)}{1-\sum_{k\in \N} \min_{t=0,1} f(j_0+2k+t)}.
$$

This construction keeps the probability of the first $j_0-1$ symbols constant as $\btheta$ varies. At and beyond $j_0$, it breaks the alphabet in blocks of size $2$ indexed by $k$, assigning the smallest of the two values of $f$ in each block to one or the other symbol, according to the component $\theta_k$ of $\btheta$ corresponding to that block. For $j_0$, we can choose any value such that $Z\geq 1$. In particular, since $\sum_{j\geq 1} f(j) > 1$, we can always choose $j_0$ such that $\sum_{j<j_0} f(j) \geq 1$. It follows that $P_\btheta(j)$'s as defined are indeed probabilities. Furthermore, for all $j$ we have that $P_\btheta(j)\leq f(j)$, and therefore $\mathcal{P} \subset \Lambda(f)$ as desired. $P_\btheta$ matches one of the values of $f$ within each block in the tail, and is almost `worst-case' in this sense.

To complete the model, let the prior $\pi$ be such that $\btheta$ is a sequence of independent identically distributed Bernoulli-$1/2$ random variables. Note that the probability multiset $\{ P_\btheta(j) : j\in \Np \}$ is the same for all $\btheta$. The only difference is in the within-block positioning in the tail, randomized by $\btheta$. It is in this sense that the parameters `dither' the tail behavior.

%%%%%%%%%%%%%%%%%%%%%%%%%%%%%%%%%%%%%%%%%%%%%%%%%%%%%%%%%%%%%%%%%%%%%%%%%%%%%%%%
%%%%%%%%%%%%%%%%%%%%%%%%%%%%%%%%%%%%%%%%%%%%%%%%%%%%%%%%%%%%%%%%%%%%%%%%%%%%%%%%
\subsection{Computing the Bayes redundancy}
%%%%%%%%%%%%%%%%%%%%%%%%%%%%%%%%%%%%%%%%%%%%%%%%%%%%%%%%%%%%%%%%%%%%%%%%%%%%%%%%
%%%%%%%%%%%%%%%%%%%%%%%%%%%%%%%%%%%%%%%%%%%%%%%%%%%%%%%%%%%%%%%%%%%%%%%%%%%%%%%%

%%% TODO make this more formal 

To proceed with the computation of the Bayes redundancy, we start with an observation: under the posterior distribution, the parameters $\theta_1,\ldots,\theta_k, \ldots$ are still independent. We first provide an intuitive argument. Given the sequence, there are two distinct possibilities per block $k$: either it is represented or it is not. If it is, then the corresponding parameter $\theta_k$ is known deterministically and none of the other parameters influence it. If it is not, then the a posteriori distribution of $\theta_k$ is its a priori distribution, because one cannot infer about it from the data, and the other parameters have no influence on this either. Therefore, given the observations, the parameters remain independent.

In order to get a formal proof of independence, it is enough to check that for each $k$, $\theta_1,\ldots,\theta_k$ are independent under the posterior distribution. A basic result in Bayesian theory asserts that the density of the posterior distribution with respect to the prior distribution is proportional to the likelihood.  Given observations $X_1,\ldots,X_n$, the likelihood at $\theta_1,\ldots,\theta_k$ can be computed using counters 
$N^0_j= \sum_{i=1}^n \IND_{X_i = j_0+2j}$, $N^1_j = \sum_{i=1}^n \IND_{X_i = j_0+2j+1}$ and $N_j =  N^0_j+N^1_j=\max(N^0_j,N^1_j)$. 
It is proportional to
\[
	\prod_{j=1}^k \left( \IND_{N_j^{{\theta}_j}=N_j} \left(P_{{\btheta}}(j_0+2j+ \theta_j)\right)^{N_j} \right) \, . 
\]
Note that the joint distribution of $N_1,\ldots, N_k$  does not depend on $\btheta$, and 
conditionally on $N_1,\ldots, N_k$, the counters  $(N^0_j)_{j\leq k}$  are independent. The likelihood is thus proportional
to a product of functions of the $\theta_j$, implying the desired independence. 

Using this observation, thanks to the chain rule for mutual information, the Bayes redundancy can be written as
\begin{eqnarray*}
I\left(\btheta, X_{1:n}\right) &=& \sum_{k\in \N} I\left(\theta_k, X_{1:n} | \theta_{1:k-1}\right)\\
&=& \sum_{k\in \N} I\left(\theta_k, X_{1:n}\right) \, .
\end{eqnarray*}
% Formally, we can capture the influence of $\btheta$ on $X_{1:n}$ by introducing a (random) sequence $\mathbf{N}=(N_k)_{k\in\N}$, defined by $N_k=\IND\big\{\exists x \in X_{1:n} : x \in \{j_0+2k, j_0+2k+1\} \big\}$, i.e. $1$ if there exists some $X_i$ in block $k$ and $0$ otherwise. Note that the distribution of $N_k$ does not depend on the value of $\btheta$, but only on the unchanging probability multiset.

By conditioning further on $\mathbf{N}$, as $N_k$  and $\theta_k$ are independent, we have for each $k$:
\begin{eqnarray*}
\lefteqn{ I\left(\theta_k, X_{1:n}\right) }\\
& = &
 I\left(\theta_k, X_{1:n} | N_k =0\right) \PROB\left(N_k=0\right) \\
&& +
 I\left(\theta_k, X_{1:n} | N_k=1\right) \PROB\left(N_k=1\right).
\end{eqnarray*}

The first term is the case when block $k$ is not represented: conditionally on $N_k=0,$ $\theta_k$ and $X_{1:n}$ are independent. Therefore, $I\left(\theta_k, X_{1:n} | N_k =0\right)=0$. The second term is when block $k$ is represented: then $\theta_k$ is known deterministically, i.e. a noiseless binary channel. Therefore, $ I\left(\theta_k, X_{1:n} | N_k=1\right)=1$ (bit). Hence,
\begin{eqnarray}
I\left(\btheta, X_{1:n}\right)
 & = &\sum_{k\in \N} \PROB(N_k=1) \nonumber \\
 & \geq & \mathbb{E}\left[K_n \right] - j_0 + 1, \label{eq:bayes-redundancy}
\end{eqnarray}
where $K_n$ denotes the number of distinct symbols in $X_{1:n}$. The inequality follows from the fact that $\sum_{k\in \N} \PROB(N_k=1)$ is the expected number of distinct symbols with values at $j_0$ or beyond. Just like $\mathbf{N}$, the distribution of $K_n$ does not depend on the value of the parameter $\btheta$. The expected number of distinct symbols when sampling from a given discrete distribution has been studied in depth in the literature \citep{GneHanPit07}, and we can use the assumptions on the tail behavior of the envelope $f$ to characterize the asymptotic behavior of its expectation.

%%%%%%%%%%%%%%%%%%%%%%%%%%%%%%%%%%%%%%%%%%%%%%%%%%%%%%%%%%%%%%%%%%%%%%%%%%%%%%%%
%%%%%%%%%%%%%%%%%%%%%%%%%%%%%%%%%%%%%%%%%%%%%%%%%%%%%%%%%%%%%%%%%%%%%%%%%%%%%%%%
\subsection{Bounding the minimax redundancy}
%%%%%%%%%%%%%%%%%%%%%%%%%%%%%%%%%%%%%%%%%%%%%%%%%%%%%%%%%%%%%%%%%%%%%%%%%%%%%%%%
%%%%%%%%%%%%%%%%%%%%%%%%%%%%%%%%%%%%%%%%%%%%%%%%%%%%%%%%%%%%%%%%%%%%%%%%%%%%%%%%

The probability multiset of Equation \eqref{eq:bayes-model} can be reindexed (using ${j'}$ instead of $j$, to make it clear it's a new indexing) as follows:
\begin{equation} \label{eq:reindex-multiset}
g({j'})= \left\{
\begin{array}{lcll}
 f({j'})/Z &\quad& \mathrm{if}\ {j'}<j_0\\
 f(2{j'}-j_0) \wedge f(2{j'}-j_0+1) &\quad& \mathrm{if}\ {j'}\geq j_0 \, .
\end{array}
\right.
\end{equation}

This new probability mass function on $\Np$ corresponds to a cumulative distribution, which we call $G$. Since the number of distinct symbols in a sequence from $G$ has the same law as that from any source in the Bayes construction, we can use it to study the expectation $\EXP K_n$. We first show how $p$ and $G$ inherit certain properties from $f$ and $F$ respectively, via the following lemma proved in Appendix \ref{app:proofs}. Recall that $m_n$ is defined as the solution of $\overline{F}(x)=x/n$, where $F$ is the smoothed envelope distribution.

\begin{lem} \label{lemma:bayes-envelope}
If $f$ is ultimately monotonically non-decreasing and $F\in \textsc{mda}(\gamma)$ with $\gamma>0$, then so are  $g$ and $G$ respectively. Furthermore, if we define $(m'_n)_n$ by $m'_n=\min\left\{ k\in \Np : \overline{G}(k)\leq k/n \right\}$, we have that $m_n / m'_n\to 2$, as $n\to \infty$.
\end{lem}

The literature on infinite urn schemes, starting with \citep{MR0216548} and surveyed in \citep{GneHanPit07}, describes tight connections between the tail behavior of the sampling distribution and the sequence $(\EXP K_n)_n$. These results establish asymptotic relationships between $\EXP K_n$, $n$, $\gamma$ and the slowly varying function of $G$. Our goal here, instead, is to relate $(\EXP K_n)_n$ and the sequence of exact thresholds $(m_n)_n$. To this effect, we prove a key result in Appendix \ref{app:threshold-properties}, Lemma \ref{thm:Kn:Mn}, which effectively bounds $K_n$ from below by $m'_n$, up to a constant factor. We state it here for clarity:
\begin{lem*}[Lemma \ref{thm:Kn:Mn} in Appendix \ref{app:threshold-properties}]
Let a distribution $G$ on $\Np$ belong to some $\textsc{mda}(\gamma), \gamma>0$, with a probability mass function that is ultimately monotonically non-increasing. Define $m'_n=\min\left\{ k\in \Np : \overline{G}(k)\leq k/n \right\}$. Then there exists a constant $\kappa'_\gamma$ and some $n_0$ (that may depend on $G$), such that for all $n\geq n_0$, the expected number of distinct symbols in a sample from $G$ satisfies
\begin{displaymath}
 \kappa'_\gamma m'_n \leq \EXP K_n \, .
\end{displaymath}
\end{lem*}

By combining Lemmas \ref{lemma:bayes-envelope} and \ref{thm:Kn:Mn}, we can assert that there exists a constant $\kappa_\gamma$ (that depends only on $\gamma$) and some $n_0$ (that may depend on $F$ more generally), such that for all $n\geq n_0$ we have:
\begin{equation} \label{eq:Kn-lowerbound}
 \kappa_\gamma m_n \leq \EXP K_n \, .
\end{equation}

We are finally in position to combine the Bayes model construction and these asymptotic characterizations to give a lower bound on the minimax redundancy. By combining Equations \eqref{eq:bayes-redundancy} and \eqref{eq:Kn-lowerbound}, and using the minimax-maximin relationship, we have established the following theorem.

\begin{thm}[\textsc{frechet minimax redundancy lower bound}] \label{thm:minimax}
Let $\Lambda(f)$ be the envelope class defined by a function $f$ that is ultimately monotonically non-increasing. If the envelope is Fr\'{e}chet with index $\gamma>0$ in the sense of Definition \ref{def:max-stable-envelope}, and $(m_n)_n$ is defined according to \eqref{eq:1}, then, for some constant $\kappa_\gamma$, for large enough $n$,
\begin{equation*}
R^+(\Lambda_n) \geq \kappa_\gamma m_n \, . %% n^{\frac{\gamma}{\gamma+1}} {j'}(n).
\end{equation*}
% If $\overline{F}_c(x) \sim x^{-1/\gamma} L(x)$, then we can always construct ${j'}$ from $L$ and $\gamma$ explicitly by first choosing $\epsilon>0$, and then letting ${j'}$ coincide at large enough $n$ with $(1-\epsilon) L^\star(n)$, as in equation \eqref{eq:L-star}.
\end{thm}

Compare this to the upper bound on this redundancy expressed in Equation \eqref{eq:minimax-upperbound} which was obtained using Theorem \ref{prop:upperbound}. According 
to recent results \citep{AchJafOrlSr14} obtained in a slight variant of our model, the logarithmic gap between lower and upper bounds for the minimax redundancy is likely to be due to the weakness of Theorem \ref{prop:upperbound} to fully capture the richness of the max-stable envelope classes.%, or the inadequacy of the maximin Bayesian construction used in this section.
%% If the latter is true, it would mean that the \textsc{etac}-code is even more powerful (adaptive rather than near-adaptive) than what Theorem \ref{thm:minimax}, and through it Theorem \ref{th:main}, seem to suggest.

% \input{etac-analysis}
%%%%%%%%%%%%%%%%%%%%%%%%%%%%%%%%%%%%%%%%%%%%%%%%%%%%%%%%%%%%%%%%%%%%%%%%%%%%%%%%
%%%%%%%%%%%%%%%%%%%%%%%%%%%%%%%%%%%%%%%%%%%%%%%%%%%%%%%%%%%%%%%%%%%%%%%%%%%%%%%%
%%%%%%%%%%%%%%%%%%%%%%%%%%%%%%%%%%%%%%%%%%%%%%%%%%%%%%%%%%%%%%%%%%%%%%%%%%%%%%%%
\section{Analysis of the \textsc{etac}-Code} \label{sec:etac-analysis}
%%%%%%%%%%%%%%%%%%%%%%%%%%%%%%%%%%%%%%%%%%%%%%%%%%%%%%%%%%%%%%%%%%%%%%%%%%%%%%%%
%%%%%%%%%%%%%%%%%%%%%%%%%%%%%%%%%%%%%%%%%%%%%%%%%%%%%%%%%%%%%%%%%%%%%%%%%%%%%%%%
%%%%%%%%%%%%%%%%%%%%%%%%%%%%%%%%%%%%%%%%%%%%%%%%%%%%%%%%%%%%%%%%%%%%%%%%%%%%%%%%

We now complete the paper by analyzing the redundancy of the \textsc{etac}-code. We start with direct bounds on the codelengths of the two strings comprising the code, the mixture encoding $C_{\textsc{m}}$ and the Elias encoding $C_{\textsc{e}}$, in terms of the data-driven threshold sequences $M_n$. These need to be related to the exact threshold $m_n$, to tie the redundancy of the code with the minimax redundancy lower bound and give precise asymptotic growth expressions. For the mixture encoding, the direct bound is sufficient upon using the distribution-free properties of the thresholds given in Appendix \ref{app:threshold-properties}. For the Elias code, further work is needed to place it in the proper form, and most of this section is dedicated to that analysis. The results presented here are combined in their final form in Theorem \ref{th:main} of Section \ref{sec:mainresult}.

%%%%%%%%%%%%%%%%%%%%%%%%%%%%%%%%%%%%%%%%%%%%%%%%%%%%%%%%%%%%%%%%%%%%%%%%%%%%%%%%
%%%%%%%%%%%%%%%%%%%%%%%%%%%%%%%%%%%%%%%%%%%%%%%%%%%%%%%%%%%%%%%%%%%%%%%%%%%%%%%%
\subsection{Codelength of the mixture encoding $C_{\textsc{m}}$}
%%%%%%%%%%%%%%%%%%%%%%%%%%%%%%%%%%%%%%%%%%%%%%%%%%%%%%%%%%%%%%%%%%%%%%%%%%%%%%%%
%%%%%%%%%%%%%%%%%%%%%%%%%%%%%%%%%%%%%%%%%%%%%%%%%%%%%%%%%%%%%%%%%%%%%%%%%%%%%%%%

The difference between the length of the progressive mixture encoding of the censored sequence can be compared with the ideal codeword length for the source output \citep[see Lemma 2 and the proof of Theorem 8 in][for details]{boucheron:garivier:gassiat:2006}:
\begin{eqnarray*}
\lefteqn{\ell (C_{\textsc{m}}) +\log \PROB_n(X_{1:n})}\\ &=& - \log \KT_{M_n+1}(\widetilde{X}_{1:n}) +\log \PROB_n(X_{1:n})\\
&& -\log Q_n(\widetilde{X}_{1:n}) + \log \KT_{M_n+1}(\widetilde{X}_{1:n})
\end{eqnarray*}
where $ \KT_{M_n+1}$ is the Krichevsky-Trofimov mixture coding probability over an alphabet of cardinality $M_n+1$. The second part of the equality is non-positive, and it follows that:
 \begin{equation*}
 \ell (C_{\textsc{m}}) +\log \PROB_n(X_{1:n}) \leq \frac{M_n+1}{2} \log(n) + 2 \, .
\end{equation*}

We can then appeal directly to the distribution-free properties of the thresholds given in Appendix \ref{app:threshold-properties}, to bound $\EXP[M_n]$ asymptotically by $m_n$. In particular, if the source belongs to a max-stable envelope class $\Lambda$, and $(m_n)$ is the corresponding exact threshold sequence as in Equation \eqref{eq:1}, we then have:
\begin{eqnarray}
\EXP\left[ \ell (C_{\textsc{m}}) +\log \PROB_n(X_{1:n}) \right]
    &\leq& \frac{\EXP[M_n]+1}{2} \log(n) + 2  \nonumber \\
    &\leq& \frac{m_n+3\sqrt{m_n}+4}{2}\log(n) + 2 \nonumber \\
    &\leq& (1/2+o_\Lambda(1))m_n \log(n) \label{eq:mixture-code}
\end{eqnarray}
where the second inequality follows from Lemma \ref{lem:trick:elizabeth} and the last inequality holds since $m_n$ grows unbounded with $n$, by Lemma \ref{lem:implict:func}.

% \begin{equation} \label{eq:mixture-code}
% \EXP\left[ \ell (C_{\textsc{m}}) +\log \PROB_n(X_{1:n}) \right] \leq (1/2+o_\Lambda(1))m_n \log(n)\, .
% \end{equation}

%%%%%%%%%%%%%%%%%%%%%%%%%%%%%%%%%%%%%%%%%%%%%%%%%%%%%%%%%%%%%%%%%%%%%%%%%%%%%%%%
%%%%%%%%%%%%%%%%%%%%%%%%%%%%%%%%%%%%%%%%%%%%%%%%%%%%%%%%%%%%%%%%%%%%%%%%%%%%%%%%
\subsection{Codelength of the Elias encoding $C_{\textsc{e}}$}
%%%%%%%%%%%%%%%%%%%%%%%%%%%%%%%%%%%%%%%%%%%%%%%%%%%%%%%%%%%%%%%%%%%%%%%%%%%%%%%%
%%%%%%%%%%%%%%%%%%%%%%%%%%%%%%%%%%%%%%%%%%%%%%%%%%%%%%%%%%%%%%%%%%%%%%%%%%%%%%%%

Over light-tailed envelope classes, the contribution of the Elias penultimate encoding of the redacted symbols to the redundancy of the \textsc{ac}-code is asymptotically negligible, relative to the mixture code length and the minimax redundancy \citep{bontemps2012adaptiveb}. The argument is transparent: when using the \textsc{ac}-code the threshold, which is the maximum, corresponds to a rank within the order statistics that is deterministic, equal to $1$, and redacted symbols are just records (excesses over maxima) of an independent sequence of identically distributed random variables. They may be analyzed using the well-established theory of records \citep[see][]{Res87}. Furthermore, the fact that envelopes have non-decreasing hazard rate considerably simplifies the analysis of extreme order statistics \citep[see][]{boutho2012}.

Over max-stable envelope classes, the analysis of the contribution of the Elias encoding to the redundancy of the \textsc{etac}-code faces new challenges. These stem from the fact that redacted symbols are not records anymore, not even $k$-th rank records for a deterministic $k$, as the threshold is determined from the data itself. Moreover, it is not straightforward to transfer properties from a sequence drawn from the smoothed envelope distribution to one drawn from a specific distribution in the envelope class. The details of the approach we follow involve tackling the problem on these fronts.

%% TODO: Add precise theorem/lemma reference for this, from the final version of the 2012 paper.
In what follows let $G$ denote our sampling distribution, which belongs to an envelope class given by $F\in \textsc{mda}(\gamma)$, $\gamma\geq 0$. As in \citep{bontemps2012adaptiveb}, the length of the Elias encoding is readily upper bounded as follows.
\begin{eqnarray*}
\EXP[\ell(C_{\textsc{e}})]
    &\leq& 2 \sum_{i=1}^{n-1} \EXP \left[ \IND\{X>M_i\} \left( \log(1+X-M_i) + \rho \right) \right]
\end{eqnarray*}
where we write a generic $X$ instead of $X_{i+1}$, because $X_{i+1}$ is always independent of $M_i$. The $\rho$ term is a parametrization choice. It contributes to the sum with a factor of $\PROB\{X>M_i\}=\EXP[\overline{G}(M_i)]$. We shortly bound the latter in Lemma \ref{lem:mean}, and meanwhile place most of our focus on bounding the logarithmic term.

We go through these general steps:
\begin{itemize}
\item For each $i$, we condition on $M_i=\u$. This reduces the problem to bounding the following `pointwise' (in the threshold) function from above:
\begin{eqnarray*}
\Sigma(\u)
    &:=& \EXP \left[ \IND\{X>M_i\} \log(1+X-M_i) | M_i=\u \right]
    % &=& \int_\u^\infty \frac{\overline{G}(x)}{1+x-\u} \mathd x
\end{eqnarray*}
Note that upon conditioning, we lose the dependence on $i$. The influence of $i$ on the total expectation is only through the distribution of $M_i$.
\item We then take a total expectation $\EXP[\Sigma(M_i)]$ for each $i$, and transfer the pointwise bounds. Since we would like to express $\EXP[\ell(C_{\textsc{e}})]$ as a function of the thresholds, we take care to relate the various bounds to $\EXP[M_i]$, and $m_i$.
\item To bound $\EXP[\ell(C_{\textsc{e}})]$, we combine the bounds for various values of $i$ in the sum:
\begin{equation} \label{eq:total-bound}
    \EXP[\ell(C_{\textsc{e}})] \leq 2\sum_{i=1}^{n-1} \left( \EXP[\Sigma(M_i)] + \rho\EXP[\overline{G}(M_i)]\right).
\end{equation}
\end{itemize}

Although each step corresponds to a simple statement, we list the results as lemmas, to cleanly delineate the proofs. We start with giving a pointwise bound on $\Sigma(\u)$.

\begin{lem} \label{lem:pointwise-bounds}
Given $\epsilon>0$ there exists $t_0$ (which will depend on both $\epsilon$ and $F$), such that for all $t>t_0$, we have:
\begin{equation} \label{eq:pointwise-unified}
\Sigma(\u) \leq \overline{G}(\u) \log(t) + (\gamma/\ln 2+\epsilon) \overline{F}(t)
\end{equation}
\end{lem}

We would now like to take the expectation $\EXP[\Sigma(M_i)]$. The only term that contributes is $\EXP[\overline{G}(M_i)]$, which can be bounded as follows.

\begin{lem} \label{lem:mean}
We have:
$$
  \EXP[\overline{G}(M_i)] \leq \frac{\EXP[M_{i+1}]}{i+1}.
$$
\end{lem}

As in the derivation of Equation \ref{eq:mixture-code} for the mixture codelength, we can now use the concentration properties of $M_i$ given by Lemma \ref{lem:trick:elizabeth}, to relate it back to $m_i$. Given $\epsilon>0$, for large enough $i$ we have:
\begin{eqnarray}
\EXP[\overline{G}(M_i)]
    &\leq& (1+\epsilon)\frac{m_{i+1}}{i+1}. \label{eq:EGUi-mi}
\end{eqnarray}

With the choice of $t=m_{i+1}$, Equations \eqref{eq:pointwise-unified} and \eqref{eq:EGUi-mi} give us that for large enough $i$:
\begin{eqnarray}
\lefteqn{\EXP[\Sigma(M_i)]}\\
    &\leq& (1+\epsilon) \frac{m_{i+1}}{i+1} \log(t) + (\gamma/\ln 2+\epsilon) \overline{F}(t) \nonumber \\
    &\leq& (1+\epsilon) \frac{m_{i+1}}{i+1} \log(m_{i+1}) + (\gamma/\ln 2+\epsilon)~ \frac{m_{i+1}}{i+1} \label{eq:ESUi-mi}  \, . 
\end{eqnarray}

Lastly, by combining these steps via Equation \eqref{eq:total-bound}, we obtain a master bound on the expected Elias codelength, valid for both Fr\'echet and Gumbel envelopes, and which we can further specialize in the Fr\'echet case. We present this in the following lemma.

\begin{lem}[\textsc{elias codelength}] \label{lem:elias-code}
Given $F\in\textsc{mda}(\gamma)$, then for all $G$ in the envelope class $\Lambda$ characterized by $F$, we have the following bounds for the Elias portion of the \textsc{etac}-code:
\begin{enumerate}[(i)]
  \item Sum bound:
      \begin{equation} \label{eq:elias-bound-sum}
      \EXP[\ell(C_{\textsc{e}})] \leq (2+o_{\Lambda}(1)) \sum_{i=1}^{n} \frac{1}{i} m_{i}\log m_{i}
      \end{equation}
  \item Integral bound:
      \begin{equation} \label{eq:elias-bound-integral}
      \EXP[\ell(C_{\textsc{e}})] \leq (2+o_{\Lambda}(1)) \int_1^n \frac{1}{t} m(t) \log m(t) \mathd t.
      \end{equation}
  \item Direct bound (for both Gumbel and Fr\'echet):
      \begin{equation} \label{eq:elias-bound-general}
      \EXP[\ell(C_{\textsc{e}})] \leq (2+o_{\Lambda}(1)) m_n \log n \log m_n.
      \end{equation}
  \item If $\gamma>0$ (only Fr\'echet):
      \begin{equation} \label{eq:elias-bound-frechet}
      \EXP[\ell(C_{\textsc{e}})] \leq (2+o_{\Lambda}(1)) m_n \log n.
      \end{equation}
\end{enumerate}
\end{lem}

%\bibliography{bigbib} 
% Generated by IEEEtranSN.bst, version: 1.13 (2008/09/30)

% \input{appendix}
%%%%%%%%%%%%%%%%%%%%%%%%%%%%%%%%%%%%%%%%%%%%%%%%%%%%%%%%%%%%%%%%%%%%%%%%%%%%%%%%
%%%%%%%%%%%%%%%%%%%%%%%%%%%%%%%%%%%%%%%%%%%%%%%%%%%%%%%%%%%%%%%%%%%%%%%%%%%%%%%%
%%%%%%%%%%%%%%%%%%%%%%%%%%%%%%%%%%%%%%%%%%%%%%%%%%%%%%%%%%%%%%%%%%%%%%%%%%%%%%%%
\appendices
%%%%%%%%%%%%%%%%%%%%%%%%%%%%%%%%%%%%%%%%%%%%%%%%%%%%%%%%%%%%%%%%%%%%%%%%%%%%%%%%
%%%%%%%%%%%%%%%%%%%%%%%%%%%%%%%%%%%%%%%%%%%%%%%%%%%%%%%%%%%%%%%%%%%%%%%%%%%%%%%%
%%%%%%%%%%%%%%%%%%%%%%%%%%%%%%%%%%%%%%%%%%%%%%%%%%%%%%%%%%%%%%%%%%%%%%%%%%%%%%%%

%%%%%%%%%%%%%%%%%%%%%%%%%%%%%%%%%%%%%%%%%%%%%%%%%%%%%%%%%%%%%%%%%%%%%%%%%%%%%%%%
%%%%%%%%%%%%%%%%%%%%%%%%%%%%%%%%%%%%%%%%%%%%%%%%%%%%%%%%%%%%%%%%%%%%%%%%%%%%%%%%
\section{Regular variation} \label{app:regular-variation}
%%%%%%%%%%%%%%%%%%%%%%%%%%%%%%%%%%%%%%%%%%%%%%%%%%%%%%%%%%%%%%%%%%%%%%%%%%%%%%%%
%%%%%%%%%%%%%%%%%%%%%%%%%%%%%%%%%%%%%%%%%%%%%%%%%%%%%%%%%%%%%%%%%%%%%%%%%%%%%%%%
Useful material concerning regular variation and applications to the analysis of extreme value theory can be found in \citep{BinGolTeu89} and in the appendix of \citep{HaaFei06}. In the current appendix, we gather some of the basic definitions and properties which we need in the paper.

\begin{dfn}[\textsc{regular variation}]\label{dfn:regular-variation}
  A measurable function $h : \mathbb{R}_+\to \mathbb{R}_+$ is \emph{regularly varying} if and only if for all $x>0$,
  $$
    \lim_{t\to \infty} \frac{h(tx)}{h(t)}
  $$
  exists as a function of $x$.
\end{dfn}

If $h$ is regularly varying, there exists some $\alpha\in \mathbb{R}_+$, such that $ \lim_{t\to \infty} \frac{h(tx)}{h(t)}=x^\alpha$, $\alpha$ is called the regular variation index of $h$, this is abbreviated as $h\in \textsc{rv}(\alpha)$. If the regular variation index is $0$, i.e. $\forall x: \frac{h(tx)}{h(t)}\to 1$, then $h$ is said to be \emph{slowly varying}.

\begin{dfn}[\textsc{extended regular variation}]\label{dfns:regular-variation-extend}
  A measurable function $h : \mathbb{R}_+\to \mathbb{R}_+$ is said to have the \emph{extended regular variation} property if and only if for all $x,y>0$,
  $$
    \lim_{t\to \infty} \frac{h(tx)-h(t)}{h(ty)-h(t)}
  $$
  exists as a function of $x,y$.
\end{dfn}

If $h$ has the extended regular variation property, there exists some $\alpha\in \mathbb{R}_+$, such that $\lim_{t\to \infty} \frac{h(tx)-h(t)}{h(ty)-h(t)} =\frac{\int_1^x u^{\alpha-1}\mathrm{d}u}{\int_1^y u^{\alpha-1}\mathrm{d}u}$, $\alpha$ is called the extended regular variation index of $f$, this is abbreviated as $f\in \textsc{erv}(\alpha)$. If $f\in \textsc{erv}(\alpha)$, then there exists an auxiliary function $a\in \textsc{rv}(\alpha)$, such that
$$
  \lim_{t\to \infty} \frac{h(tx)-h(t)}{a(t)} = \int_1^x u^{\alpha-1} \mathrm{d}u\, .
$$
For $\alpha>0$, we have $\textsc{erv}(\alpha)=\textsc{rv}(\alpha)$. If $f\in \textsc{erv}(0)$ and $\lim_{t\to \infty} h(t) = \infty$,
then $f\in \textsc{rv}(0)$.

%% TODO: Potter is not referenced. Keep?
A fundamental result in regular variation theory asserts that the convergence in Definitions \ref{dfn:regular-variation} and \ref{dfns:regular-variation-extend} is locally uniform over compact sets. Potter's inequalities \citep[See][for a proof]{BinGolTeu89,HaaFei06} provide us with a useful quantitative formulation of this result.
%===============================================================================
\begin{lem}[\textsc{potter's inequalities}] \label{lemma:potters-inequalities}
  Let $f\in \textsc{rv}(\alpha)$ then for all $\epsilon,\delta>0$, there exists $t_0(\epsilon,\delta)$ such that for all $t,x>0$ such that $\min(t,tx)>t_0(\epsilon,\delta)$,
  $$
    \left| \frac{h(tx)}{h(t)}-x^\alpha\right| \leq \epsilon x^\alpha\max(x^\delta,x^{-\delta}) \, .
  $$
\end{lem}
%===============================================================================

Another core result is Karamata's integration theorem \citep[See][Theorem B.1.5]{HaaFei06}, which intuitively tells us that regularly varying functions can be integrated roughly like their defining monomials.
%===============================================================================
\begin{thm}[\textsc{karamata's integration}] \label{thm:karamata}
 Let $h$ be regularly varying with index $\alpha$. Then:
 \begin{itemize}
  \item[-] There exists $t_0>0$ such that $h(t)$ is positive and locally bounded for $t\geq t_0$.
  \item[-] If $\alpha\geq-1$ then:
    $$
      \lim_{t\to\infty} \frac{t h(t)}{\int_{t_0}^t h(s)\mathd s} = \alpha+1.
    $$
  \item[-] If $\alpha\leq-1$ and $\int h(s) \mathd s < \infty$, then:
    $$
      \lim_{t\to\infty} \frac{t h(t)}{\int_{t}^\infty h(s)\mathd s} = -\alpha-1.
    $$
 \end{itemize}
\end{thm}
%===============================================================================

A concept that proves very useful when relating various regularly varying functions is that of De Bruijn conjugacy, which is effectively a notion of asymptotic inversion of slowly varying functions.
%===============================================================================
\begin{thm}[\textsc{de bruijn conjugacy}]\citep[Proposition 1.5.15]{BinGolTeu89} \label{the:debruijn}
  Let $L\in \textsc{rv}(0)$, then there exists a function $L^*\in \textsc{rv}(0)$ such that $L^*(x)L(xL^*(x))\rightarrow 1$ and $L(x)L^*(xL(x))\rightarrow 1$ as $x \to \infty$. Any function satisfying these two relations is asymptotically equivalent to $L^*$. The functions $(L,L^*)$ are said to form a pair of De Bruijn conjugates.
\end{thm}
\section{Properties of the exact and empirical thresholds} \label{app:threshold-properties}
%%%%%%%%%%%%%%%%%%%%%%%%%%%%%%%%%%%%%%%%%%%%%%%%%%%%%%%%%%%%%%%%%%%%%%%%%%%%%%%%
%%%%%%%%%%%%%%%%%%%%%%%%%%%%%%%%%%%%%%%%%%%%%%%%%%%%%%%%%%%%%%%%%%%%%%%%%%%%%%%%
\label{sec:thresholds}
At the heart of the upper bound on the redundancy of the \textsc{etac}-code derived in Section \ref{sec:etac-analysis} is an understanding of the connection between the empirical threshold $M_n$ constructed from data coming from a specific distribution in the envelope class and the exact threshold $m_n$ induced by the envelope distribution. It turns out that these results can be derived without particular assumptions on the distributions, and we present them in this appendix. After this, we tackle the task of relating $m_n$ to the number of distinct symbols $K_n$, which is the quantity that governs the minimax lower bound presented in Section \ref{sec:minimax}. There, we use more closely the max-stability property of the envelope distribution.

%%%%%%%%%%%%%%%%%%%%%%%%%%%%%%%%%%%%%%%%%%%%%%%%%%%%%%%%%%%%%%%%%%%%%%%%%%%%%%%%
\subsection{Distribution-free properties}
%%%%%%%%%%%%%%%%%%%%%%%%%%%%%%%%%%%%%%%%%%%%%%%%%%%%%%%%%%%%%%%%%%%%%%%%%%%%%%%%

When analyzing the \textsc{ac}-code of \citep{Bon11}, one needs to work only with extreme order statistics, that is statistics of constant order, and in particular the threshold there is the maximum (order $1$). The fact that the threshold $M_n$ of the \textsc{etac}-code is effectively equal to $X_{M_n,n}$, which is an \emph{intermediate order statistic} (that is $M_n \to \infty$ while $M_n/n \to 0 $ in probability) rather than an extreme may seem to add difficulty to the analysis of the code. Nevertheless, the fluctuations of $M_n$ around its mean value $\EXP M_n $ can be bounded in a surprisingly simple way. Moreover, this concentration result does not depend on any assumption regarding the distribution of the sample. The fluctuation bounds only depend on the fact that $M_n$ is a function of independent random variables that does not depend too much on any of them.

When working with random variables that can be expressed as functions of other random variables, self-boundedness is a property that can simplify the derivation of moments bounds and concentration properties. We give here the basic definition of self-bounded random variables \citep*[refer to][Chapters 3 and 6]{boluma13}.

\begin{dfn}[\textsc{self-boundedness}] \label{dfn:self-bounded}
A non-negative random variable $Z=g(X_1,\ldots,X_n)$, that is a function of $n$ other variables $X_1,\ldots,X_n$, is called \emph{self-bounded} if there exists a collection of measurable functions $(g_i)_{i\leq n}$, such that letting $Z_i= g_i(X_1,\ldots,X_{i-1},X_{i+1},\ldots,X_n)$, we have
\begin{displaymath}
  \begin{array}{ll}
    0\leq Z - Z_i \leq 1 & \text{for each } i \leq n \\
    \sum_{i=1}^n \left( Z-Z_i \right)\leq Z \, . &
  \end{array}
\end{displaymath}
\end{dfn}

The next lemma establishes self-boundedness and uses it to assert that whatever the sampling distribution, $M_n$ has ``sub-Poissonian'' tails.
%===============================================================================
\begin{lem}\label{lem:mnmeetsefrons-stein}
Let $X_{1,n}\geq \ldots \geq X_{n,n}$ be the order statistics of an i.i.d. sample, let $M_n=\min(n,\inf\{k\colon X_{k,n}\leq k\})$, then:
\begin{enumerate}[(i)]
  \item\label{item:self-bounded} $M_n$ is a self-bounded random variable, as in Definition \ref{dfn:self-bounded}.
  \item\label{item:moments} We have the moment bounds:
  $$
    \operatorname{var} (M_n) \leq \mathbb{E} M_n
  $$
  and for all $\lambda \in \mathbb{R},$
  $$
    \log \EXP \left[ \mathe^{\lambda (M_n-\EXP M_n)}\right] \leq \EXP M_n \left( \mathe^\lambda -\lambda -1\right) \, .
  $$
  \item\label{item:concentration} For all $t>0$, we have:
  $$
    \mathbb{P} \left\{ M_n-\EXP M_n\geq t\right\} \leq \exp\left(- \frac{t^2}{2(\mathbb{E}M_n +t/3)} \right) \, ,
  $$
  and
  $$
    \mathbb{P} \left\{ M_n-\EXP M_n\leq -t \right\} \leq \exp\left(- \frac{t^2}{2(\mathbb{E}M_n)} \right) \, .
  $$
 \end{enumerate}
\end{lem}
%===============================================================================
\begin{proof}[Proof of Lemma \ref{lem:mnmeetsefrons-stein}]
Let $Z= M_n(X_1,\ldots,X_n)$, and for each $i=1,\ldots,n$, let
\begin{displaymath}
Z_i = \inf\left\{ M_n(X_1,\ldots,X_{i-1},x'_i,X_{i+1},\ldots,X_n)~:~x'_i \in \mathbb{N}\right\} \, .
\end{displaymath}
As $M_n$ is non-decreasing with respect to the product order on the sample, in order to have
$$
  M_n(X_1,\ldots,X_{i-1},x'_i,X_{i+1},\ldots,X_n)\leq M_n\, ,
$$
it is necessary to choose $x'_i\leq X_i.$

If $X_i \leq X_{M_n ,n}$, choosing $x'_i$ smaller than $X_i$ does not modify the $M_n$ largest order statistics
and the value of $M_n$.

If $X_i > X_{M_n ,n}$ by choosing $x'_i<X_{M_n,n}$, we obtain
\begin{eqnarray*}
 \lefteqn{ M_n(X_1,\ldots,X_{i-1},x'_i,X_{i+1},\ldots,X_n)}\\
&=&
    \begin{cases} M_n-1 &
      \text{ if } X_{M_n,n}\leq M_n-1 \\ %M_n(X_1,\ldots,X_{i-1},x'_i,X_{i+1},\ldots,X_n)=
      M_n & \text{ otherwise.}
    \end{cases}
\end{eqnarray*}
Hence $0\leq Z -Z_i \leq 1$ for all $1\leq i \leq n$, and $\sum_{i=1}^n (Z-Z_i)\leq M_n$.
This establishes \eqref{item:self-bounded}. Then \eqref{item:moments} and \eqref{item:concentration} follow from Corollary 3.7 and Theorem 6.12 in \citep*{boluma13}.
\end{proof}

As usual, concentration inequalities need to be complemented by bounds on expectations. Fortunately, the expected value of $M_n$ can again be related to $m_n$ without any distributional assumptions.
%===============================================================================
\begin{lem}\label{lem:trick:elizabeth}
% Assume $F$ (and hence $U=F^\leftarrow(1-1/.)$) is continuous and strictly increasing, with $U(\infty)=\infty$,
Let $G$ be a source that belongs to an envelope class $\Lambda$ defined by a smoothed distribution $F$. Recall the definition of the exact threshold sequence $m_n$ as the solution of $\overline{F}(x)=x/n$ and of the threshold sequence $M_n =\min \left( n, \left\{k\;:\;X_{k,n} \leq k \right\}\right)$, where $X_{1,n}\geq X_{2,n} \geq \ldots \geq X_{n,n}$ are the order statistics of an  $n$-length sequence from $G$. Then, for all $n$, we have:
\begin{displaymath}
\EXP M_n \leq m_n + 3\sqrt{m_n} + 3\, .
\end{displaymath}
\end{lem}
%===============================================================================
\begin{proof}[Proof of Lemma \ref{lem:trick:elizabeth}]
 %If $1<M_n<n$, $X_{M_n,n}\leq M_n < 1+X_{M_n-1,n}$. If $M_n >m_n+1$
% \begin{displaymath}
% \frac{M_n-1}{n} \leq \overline{F}_n(M_n-1) \leq \overline{F}_n(m_n) \, .
% \end{displaymath}
% As we assume $m_n>1$ and $\overline{F}(n)\leq 1/2$,
% \begin{eqnarray*}
% \EXP M_n & \leq & \EXP \left[ M_n \IND_{M_n\leq m_n+1 }\right] + \EXP \left[ M_n \IND_{n>M_n> m_n+1 }\right] + n \PROB\{ M_n =n\} \\
% & \leq & m_n +1+ \EXP \left[ \left( 1+ n \overline{F}_n(m_n)\right) \IND_{n>M_n>1+ m_n }\right] + n \overline{F}(n)^n \\
% & \leq & m_n + 1+\EXP \left[ \left( 1+ n \overline{F}_n(m_n)\right)\right] + n \overline{F}(n)^n \\
% & \leq & m_n + n \overline{F}(m_n) + 3 \\
% & = & 2 m_n+3 \, .
% \end{eqnarray*}

We prove a stronger, two-sided, inequality involving the analog to $m_n$ defined directly for $G$ instead of $F$:
$$
  m'_n = \min\left\{ k \colon \overline{G}(k)\leq k/n \right\}.
$$
In particular, we show that:
\begin{displaymath}
 m'_n - 3 \sqrt{m'_n} - 2\leq \EXP M_n \leq m'_n+3\sqrt{m'_n}+3.
\end{displaymath}
The assertion of the lemma then follows from the fact that $m'_n \leq m_n$, which is a direct consequence of the fact that if $G$ is in the envelope class defined by $F$, then $\overline{G} \leq \overline{F}$.

We compare the expectations of $M_n$ and $m'_n$ with the following steps. Let $\MED[M_n]$ be a median of the distribution of $M_n$,  that is $\MED[M_n]$ satisfies $\PROB \{M_n \leq \MED[M_n] \}\geq 1/2$ and $\PROB\{ M_n \geq \MED[M_n]\}\geq 1/2$. If we establish concentration bounds to quantify $\PROB\{ M_n \leq a \}$ and $\PROB\{ M_n \geq b \}$ for suitable $a,b>0$ in the tail of $M_n$, we may choose $a$ and $b$ such that these probabilities drop below $1/2$. We can then deduce that:

$$
  a \leq \MED[M_n] \leq b
$$
To move from the median to the mean, note that by the L\'evy-Mallows inequality and point \eqref{item:moments} of Lemma \ref{lem:mnmeetsefrons-stein}, we have
$$
  |\MED[M_n] -\EXP M_n|\leq \sqrt{\operatorname{var}(M_n)}\leq \sqrt{\EXP M_n},
$$
from which we can directly deduce
\begin{eqnarray*}
\lefteqn{\MED[M_n] - \sqrt{\MED[M_n]}} \\
&\leq &\EXP M_n \\
&\leq &\MED[M_n] + 1 + \sqrt{\MED[M_n]},
\end{eqnarray*}
and thus
\begin{equation} \label{eq:mean-median-bounds}
a - \sqrt{b} \leq \EXP M_n \leq b + \sqrt{b} + 1.
\end{equation}

It remains to establish the concentration bounds, and to obtain explicit values for $a$ and $b$. For this, we compare the events of interest to binomial tails, namely the empirical tail count:
$$
  n \overline{G}_n(x)=\sum_{i=1}^{n} \IND\{X_i > x\}.
$$

Note the following properties:
\begin{enumerate}[(i)]
   \item $n \overline{G}_n$ is a non-increasing random function of $x$ and, for fixed $x$, has a binomial distribution of parameters $n$ and $\overline{G}(x)$.
   \item If $\overline{G}(x)\leq q$ and $Z\sim \mathrm{Binomial}(n,q)$, then $\PROB\{ n\overline{G}_n(x) \geq b \} \leq \PROB\{Z \geq b\}$.
   \item If $\overline{G}(x)\geq p$ and $Z\sim \mathrm{Binomial}(n,p)$, then $\PROB\{ n\overline{G}_n(x) \leq a \} \leq \PROB\{Z \leq a\}$.
   \item We have $n\overline{G}_n(M_n-1) \geq M_n-1$. It follows that when $x\leq M_n-1$ we have $n\overline{G}_n(x) \geq x$.
   \item We have $n\overline{G}_n(M_n)\leq M_n$. It follows that, when $x\geq M_n$ we have $n\overline{G}_n(x) \leq x$.
\end{enumerate}

The first three properties are evident. The last two make it clear that $M_n$ is effectively an empirical version of $m'_n$.  To establish (iv) $n\overline{G}_n(M_n-1)\geq M_n-1$: all statistics from $X_{1,n}$ to $X_{M_n-1,n}$ are no less than $X_{M_n-1,n}$; the latter is itself greater than $M_n-1$, by the definition of $M_n$. To establish (v) $n\overline{G}_n(M_n)\leq M_n$: no order statistic beyond $X_{M_n,n}$ can be strictly greater than $X_{M_n,n}$; the latter itself either is no greater than $M_n$ or is so but $M_n=n$, by the definition of $M_n$; in both cases the claim remains valid. % One of these inequalities could perhaps be sharpened to a strict inequality.

Let $t>0$. When $M_n \geq m'_n +1+t$, we have that $M_n - 1 \geq m'_n + t$ (this is why we need the extra $+1$). It follows from (iv) that $n\overline{G}_n(m'_n + t) \geq m'_n + t$. Then by the non-increasing property we also have $n \overline{G}_n(m'_n)\leq m'_n+t$. By the definition of $m'_n$, we have $\overline{G}(m'_n)\leq m'_n/n$. Let $Z_1\sim \mathrm{Binomial}(n,m'_n/n)$, then by (ii) we have:
\begin{eqnarray*}
\PROB\{M_n \geq m'_n +1+t \} & \leq &\PROB\{n\overline{G}_n(m'_n) \geq m'_n + t\}\\
 &\leq& \PROB\{Z_1\geq m'_n + t\}.  
\end{eqnarray*}
On the other hand, when $M_n \leq m'_n -1-t$, it follows from (v) that $n\overline{G}_n(m'_n-1-t) \leq m'_n-1-t$. (The case $m'_n=1$ becomes pathological in what follows, but since it allows for any choice of $t$ to yield a vacuous lower bound of the median, we ignore it here.) By the non-increasing property, we also have that $n\overline{G}_n(m'_n-1) \leq m'_n-1-t$. By the definition of $m'_n$, we have that $\overline{G}(m'_n-1) \geq (m'_n-1)/n$ (this is where we need the extra $-1$). This time, let $Z_2\sim \mathrm{Binomial}(n,(m'_n-1)/n)$, then by (iii) we have:
\begin{eqnarray*}
  \PROB\{M_n \leq m'_n -1-t\} &\leq &\PROB\{n\overline{G}_n(m'_n-1) \leq m'_n-1-t\}\\
 &\leq &\PROB\{Z_2\leq m'_n-1-t\}.
\end{eqnarray*}
We recall now Bernstein's inequalities to bound the tail of binomial distributions. In particular, we have:
$$
\PROB\{ Z_1 \geq m'_n + t \} \leq \exp \left[ -\frac{3}{8} \left( t \wedge \frac{t^2}{m'_n} \right) \right],
$$
and
$$
\PROB\{ Z_2 \leq m'_n -1-t\} \leq \exp \left( -\frac{1}{2} \frac{(t+1)^2}{m'_n-1} \right).
$$

It is then easy to verify that the choice of $t=2\sqrt{m'_n}$ sets both of these bounds below the desired level of $\frac{1}{2}$ for all values of $m'_n$. Therefore, we can bound $\MED[M_n]$ by $a=m'_n-1-t=m'_n-2\sqrt{m'_n}-1$ from below and by $b=m'_n+1+t= m'_n+2\sqrt{m'_n}+1=(\sqrt{m'_n}+1)^2$ from above, and use these quantities in Equation \eqref{eq:mean-median-bounds} to bound $\EXP[M_n]$. The constants claimed in the lemma follow immediately.
\end{proof}

\subsection{Distribution-dependent properties}
%%%%%%%%%%%%%%%%%%%%%%%%%%%%%%%%%%%%%%%%%%%%%%%%%%%%%%%%%%%%%%%%%%%%%%%%%%%%%%%%

We now describe a general connection between $M_n$ and the number of distinct symbols $K_n$, that is the expected size of the empirical alphabet. From the very definition of $M_n$, if $M_n<n$, we have $K_n\leq 2 M_n$. Indeed, as $X_{M_n,n}<M_n$ there are no more than $M_n$ distinct symbols not larger than $X_{M_n,n}$ and there are at most $M_n$ distinct symbols larger than $X_{M_n,n}$. Hence, whatever the sampling distribution,
\begin{displaymath}
 \EXP K_n \leq 2 \EXP M_n \, .
\end{displaymath}

As we use $\EXP K_n$ in the lower bound on minimax redundancy, we actually need an inequality in the other direction. We now establish this under distributional assumptions.

%===============================================================================
\begin{lem} \label{thm:Kn:Mn}
Let a distribution $G$ on $\Np$ belong to some $\textsc{mda}(\gamma), \gamma>0$, with a probability mass function that is ultimately monotonically non-increasing. Let $(m'_n)_n$ be defined as $m'_n=\min\left\{ k\in \Np : \overline{G}(k)\leq k/n \right\}$. Then there exists a constant $\kappa'_\gamma$ and some $n_0$ (that may depend on $G$), such that for all $n\geq n_0$, the expected number of distinct symbols in a sample from $G$ satisfies
\begin{displaymath}
 \kappa'_\gamma m'_n \leq \EXP K_n \, .
\end{displaymath}
\end{lem}
%===============================================================================

\begin{proof}

Let $g$ denote the probability mass function corresponding to $G$. The regular variation property of $G$ then passes in a straightforward way to $g$ via so-called Tauberian theorems. In particular, recalling that we can write $\overline{G}(x) = x^{-1/\gamma} L(x)$, a simple adaptation of Theorem 1.7.2 of \cite{BinGolTeu89} shows that as $j\to\infty$:
$$
 g(j) \sim \frac{\overline{G}(j)}{\gamma j} \, .
$$

Given $\beta\geq 1$, let $k_n=\beta m'_n$ be a dilation of the threshold $m'_n$, which we will choose appropriately. Note that as $n\to\infty$, we also have that both $k_n, m'_n\to\infty$. Given $\epsilon>0$, we shall choose $n_0$ (which may depend on all of $\epsilon$, $\beta$ and $G$), such that for all $n>n_0$, several assertions hold true. In particular, for all $j\geq m'_n$, we have:
\begin{enumerate}[(i)]
 \item $g(j)$ is monotonically non-increasing (by assumption). \label{item:monotonic}
 \item $\frac{\overline{G}(j)}{(1+\epsilon)\gamma j} \leq g(j) \leq \frac{(1+\epsilon)\overline{G}(j)}{\gamma j}$ (by the Tauberian theorem). \label{item:tauberian}
 \item $\frac{\overline{G}(j)}{j} \geq \frac{1}{(1+\epsilon)} \frac{\overline{G}(j-1)}{j-1}$ (by regular variation limits over compact intervals, cf. Potter's inequalities in Lemma \ref{lemma:potters-inequalities}). \label{item:rv1}
 \item $\overline{G}(\beta m'_n) \leq (1+\epsilon) \beta^{-1/\gamma} \overline{G}(m'_n)$ (by regular variation). \label{item:rv2}
 \item $\frac{\overline{G}(\beta m'_n)}{m'_n} \geq \frac{1}{(1+\epsilon)} \beta^{-1/\gamma} \frac{\overline{G}(m'_n-1)}{m'_n-1} $ (by regular variation and similarly to \eqref{item:rv1} above). \label{item:rv3}
\end{enumerate}

% By recalling the definition of $m'_n$, this suggests that symbols larger than $m'_n$ have probability mass smaller than $\frac{1}{\gamma n}$. The number of occurrences of each of them in a sample of length $n$ is thus stochastically dominated by a Binomial random variable with parameters $(n,\frac{1}{\gamma n}).$ In what follows we make this intuition concrete.

Let $L_n$ be the number of distinct symbols in $X_1,\ldots,X_n$ that are larger than $k_n$. Then, when $n$ is beyond $n_0$, we have:
\begin{eqnarray*}
\EXP L_n
    &=& \sum_{k>k_n} \left( 1-(1-g(k))^n\right) \\
    & \geq & \sum_{k> k_n} ng(k)\left( 1-\frac{ng(k)}{2}\right) \\
    & \geq & n \overline{G}(k_n) \left(1-\frac{ng(k_n)}{2}\right) \\
    & \geq & n \overline{G}(k_n) \left(1-(1+\epsilon)\frac{n\overline{G}(k_n)}{2\gamma k_n}\right)
\end{eqnarray*}
where the first line is exact, the second step is an approximation, the third and fourth steps use assertions \eqref{item:monotonic} and \eqref{item:tauberian} respectively.

If $\gamma\geq 1$, we can simply set $\beta=1$, which would give us
\begin{eqnarray*}
\EXP L_n
  &\geq& n \overline{G}(m'_n) \left(1-(1+\epsilon)\frac{n\overline{G}(m'_n)}{2\gamma m'_n}\right) \\
  &\geq& n \frac{1}{(1+\epsilon)}m'_n \frac{\overline{G}(m'_n-1)}{m'_n-1}\left(1-(1+\epsilon)\frac{n\overline{G}(m'_n)}{2\gamma m'_n}\right) \\
  &\geq& \frac{m'_n}{1+\epsilon}\left(1-\frac{1+\epsilon}{2}\right),
\end{eqnarray*}
where the first step is a substitution, the second step uses assertion \eqref{item:rv1}, and the last step uses the fact that $\gamma\geq 1$ and the definition of $m'_n$, which implies that $G(m'_n)/m'_n\leq 1/n$ whereas $G(m'_n-1)/(m'_n-1) > 1/n$.

If $\gamma\leq 1$, if we attempt the above we end up with a lower bound that may be negative and thus vacuous. We remedy the problem by choosing $\beta$ appropriately. We have:
\begin{eqnarray*}
\lefteqn{\EXP L_n}\\
  &\geq& n \overline{G}(\beta m'_n) \left(1-(1+\epsilon)\frac{n\overline{G}(\beta m'_n)}{2\gamma \beta m'_n}\right) \\
  &\geq& n \frac{1}{(1+\epsilon)} \beta^{-\tfrac{1}{\gamma}} m'_n \frac{\overline{G}(m'_{n-1})}{m'_{n-1}} \left(1-(1+\epsilon)\beta^{-\tfrac{1}{\gamma}}\frac{n\overline{G}(m'_n)}{2\gamma \beta m'_n}\right) \\
  &\geq& \beta^{-\tfrac{1}{\gamma}} \frac{m'_n}{1+\epsilon}\left(1-\beta^{-\tfrac{1}{\gamma}-1}\frac{1+\epsilon}{2\gamma}\right),
\end{eqnarray*}
where now the second step uses assertions \eqref{item:rv2} and \eqref{item:rv3}, and the last step uses again the definition of $m'_n$. Therefore, we may choose $\beta = \gamma^{-\frac{\gamma}{\gamma+1}}$ when $\gamma<1$ to obtain the same functional form of the lower bound when $\gamma\geq 1$, up to a constant factor. We combine these two cases by letting $k_n = \left(1\vee \gamma^{-\frac{\gamma}{\gamma+1}} \right)$ and choosing $\epsilon=1/3$, to obtain:
\begin{eqnarray*}
\EXP L_n
  &\geq& \left(1\wedge \gamma^{\frac{1}{\gamma+1}} \right) \frac{m'_n}{1+\epsilon}\left(1-\frac{1+\epsilon}{2}\right) \\
  &\geq& \frac{1}{4} \left(1\wedge \gamma^{\frac{1}{\gamma+1}} \right) m'_n.
\end{eqnarray*}

This bound is sufficient to complete the lemma, since $\EXP K_n \geq
\EXP L_n$. We can try to improve it, by incorporating symbols below $m'_n$. However, without further assumptions, we cannot do so. One thing we can do, in the $\gamma<1$ case, is to smooth out this bound by incorporating symbols between $m'_n$ and $k_n$. Let $S_n$ be the number of distinct symbols $k\leq k_n$ in $X_1,\ldots,X_n$, with the same choice of $k_n$ and for $n\geq n_0$. We then have:
\begin{eqnarray*}
 \EXP S_n & = & \sum_{k\leq k_n} \left( 1-(1-g(k))^n\right)
 \IND_{g(k)>0}\\
& \geq & \sum_{m'_n < k\leq k_n} \left( 1-(1-g(k_n))^n\right) \\
& \geq & (k_n-m'_n) \left(1-\mathe^{-1/(1+\epsilon)^2}\right) \\
& \geq & \frac{1}{4} (k_n-m_n).
\end{eqnarray*}
Here, the first line is exact. The second step uses assertion \eqref{item:monotonic} both to bound the probabilities and to insure their positivity (indeed, if $g(j)=0$ for some $j\geq m'_n$, then it is so beyond that by monotonicity, which contradicts the regular variation property at infinity). The third step uses an approximation, in addition to assertions \eqref{item:tauberian} and \eqref{item:rv3}, and the definitions of $m'_n$ and $\beta$. The last step is an arbitrary (not necessary the best) choice of $\epsilon$. This bound is zero if $\gamma\geq 1$. If $\gamma < 1$, however, we have:
\begin{displaymath}
 \EXP S_n \geq \frac{m_n}{4} \left(\gamma^{-\gamma/(\gamma+1)} -1\right) \, .
\end{displaymath}

Combining $L_n$ and $S_n$, we can write:
\begin{eqnarray*}
 \EXP K_n &= &\EXP L_n + \EXP S_n \\ &\geq &\frac{m_n}{4} \left(1\wedge
 \gamma^{1/(\gamma+1)}  + \left(\gamma^{-\gamma/(\gamma+1)} -1
 \right)_+ \right) \, . \qedhere
\end{eqnarray*}

\end{proof}

%%%%%%%%%%%%%%%%%%%%%%%%%%%%%%%%%%%%%%%%%%%%%%%%%%%%%%%%%%%%%%%%%%%%%%%%%%%%%%%%
%%%%%%%%%%%%%%%%%%%%%%%%%%%%%%%%%%%%%%%%%%%%%%%%%%%%%%%%%%%%%%%%%%%%%%%%%%%%%%%%
\section{Proofs of Lemmas in the Main Text} \label{app:proofs}
%%%%%%%%%%%%%%%%%%%%%%%%%%%%%%%%%%%%%%%%%%%%%%%%%%%%%%%%%%%%%%%%%%%%%%%%%%%%%%%%
%%%%%%%%%%%%%%%%%%%%%%%%%%%%%%%%%%%%%%%%%%%%%%%%%%%%%%%%%%%%%%%%%%%%%%%%%%%%%%%%

%%%%%%%%%%%%%%%%%%%%%%%%%%%%%%%%%%%%%%%%%%%%%%%%%%%%%%%%%%%%%%%%%%%%%%%%%%%%%%%%
\subsection*{\underline{Proof of Lemma \ref{lem:implict:func}}} \label{sec:proof-lemma-implicit}
%%%%%%%%%%%%%%%%%%%%%%%%%%%%%%%%%%%%%%%%%%%%%%%%%%%%%%%%%%%%%%%%%%%%%%%%%%%%%%%%
\begin{proof}[\color{white}]
\noindent \vspace{-12pt}
\begin{enumerate}[(i)]
\item For sufficiently large $t$, $U(t/1)-1>0$, and as $x\mapsto U(t/x)-x$ decreases continuously to $-\infty$ on $[1,\infty)$, there exists some $x=m(t)$ such that
$U(t/x)-x=0.$ Hence, the function $m$ is defined over $(U^{-1}(1),\infty).$ If $U^{-1}(1)<t<t'$,
$U(t'/m(t))> U(t/m(t))=m(t)$, hence $m(t')> m(t)$.

\item The continuous differentiability of $m$ over $(U^{-1}(1),\infty)$ is a consequence of the implicit function theorem (see \citep*{trench2003introduction}). Moreover, the derivative of $m$ satisfies:
\begin{displaymath}
 m'(t) = \frac{U'\left({t}/{m(t)}\right)}{({t}/{m(t)})U'\left({t}/{m(t)}\right)+m(t)} \, .
\end{displaymath}
%which reveals that $m$ is increasing.

\item Assume on the contrary that $m$ is upper bounded by $B<\infty,$ then $U(t/B)\leq U(t/m(t))\leq B$ for $t\in (U^{-1}(1),\infty).$ As $U(t/B)$ tends to infinity, we obtain a contradiction. Assume now that $t/m(t)$ is upper bounded by $C<\infty,$ then $m(t)=U(t/m(t))\leq C$, we obtain another contradiction.

\item As $U\in \textsc{erv}(\gamma), \gamma\geq 0\text{ and }U(\infty)=\infty,$ $U\in \textsc{rv}(\gamma)$.
Then the function $L(t)= t^{-\gamma}U(t)$
is slowly varying. The definition of $m$ translates into
\begin{displaymath}
 m(t) = \frac{t^\gamma}{m(t)^\gamma} L\left( \frac{t}{m(t)}\right) \, ,
\end{displaymath}
or equivalently
\begin{displaymath}
 1 = \frac{t^\gamma}{m(t)^{1+\gamma}} L \left( \left(t \frac{t^\gamma}{m(t)^{1+\gamma}}\right)^{1/(1+\gamma)}\right)
\end{displaymath}
The function $L_1(t) = L(t^{1/(1+\gamma)})$ is slowly varying, hence the function $L_1^* \colon t\rightarrow \tfrac{t^\gamma}{m(t)^{1+\gamma}} $ appears as its De Bruijn conjugate, as such it is a slowly varying function. One line of computation reveals that $m$ is regularly varying with index $\gamma/(\gamma+1)$ and
\begin{displaymath}
 m(t) \sim t ^{\gamma/(\gamma+1)} /\left( L_1^*(t)\right)^{1/(1+\gamma)} \, . \qedhere
\end{displaymath}
\end{enumerate}
\end{proof}

%%%%%%%%%%%%%%%%%%%%%%%%%%%%%%%%%%%%%%%%%%%%%%%%%%%%%%%%%%%%%%%%%%%%%%%%%%%%%%%%
\subsection*{\underline{Proof of Lemma \ref{lemma:bayes-envelope}}}
%%%%%%%%%%%%%%%%%%%%%%%%%%%%%%%%%%%%%%%%%%%%%%%%%%%%%%%%%%%%%%%%%%%%%%%%%%%%%%%%
\begin{proof}[\color{white}]
\noindent
Since $f$ is ultimately monotonically non-decreasing, it immediately follows that the same is true for $g$. We focus therefore on showing that $G\in \textsc{mda}(\gamma)$. For this, we sandwich $\overline{G}$ by a scaled version of $\overline{F}$.

Given $\epsilon>0$, then for $k$ large enough, we have:
\begin{eqnarray}
 \overline{G}(k)
    &=& \sum_{{j'}>k} g({j'}) = \sum_{{j'}>k} f(2{j'}-j_0) \wedge f(2{j'}-j_0+1) \nonumber \\
    &\leq& \sum_{{j'}>k} \frac{1}{2} \left[f(2{j'}-j_0) + f(2{j'}-j_0+1)\right] \nonumber \\
    &=& \frac{1}{2} \overline{F}(2k-j_0) \leq \frac{(1+\epsilon)}{2} \overline{F}(2k), \label{eq:G-lowerbound}
\end{eqnarray}
where we have simply used the fact that the minimum lies below the average and the regular variation property of $\overline{F}$, with a slack of $1+\epsilon$.

Since $F$ is regularly varying with index $-1/\gamma$, by a simple adaptation of Theorem 1.7.2 of \cite{BinGolTeu89} (cf. the proof of Lemma \ref{thm:Kn:Mn} for a full relationship), so is $f$ with index $-1/\gamma-1$. In particular, it follows from this that $f(j+1)/f(j)\to 1$. Given $\delta>0$, we thus have that for ${j'}$ large enough:
$$
\frac{f(2{j'}-j_0) \wedge f(2{j'}-j_0+1)}{f(2{j'}-j_0) + f(2{j'}-j_0+1)} > \frac{1}{2+\delta}.
$$

Using this observation and the same steps above, we have that for ${j'}$ large enough:
\begin{eqnarray}
 \overline{G}(k)
    &>& \sum_{{j'}>k} \frac{1}{2+\delta} \left[f(2{j'}-j_0) + f(2{j'}-j_0+1)\right] \nonumber \\
    &=& \frac{1}{2\sqrt{1+\epsilon}} \overline{F}(2k-j_0) \geq \frac{1}{2(1+\epsilon)} \overline{F}(2k), \label{eq:G-upperbound}
\end{eqnarray}
where we choose the $\delta$ of the ratio test appropriately to get $2+\delta=2\sqrt{1+\epsilon}$, to compound its effect with the regular variation slack of $\sqrt{1+\epsilon}$ given to $\overline{F}$.

From the sandwiching offered by the two bounds of Equations \eqref{eq:G-lowerbound} and \eqref{eq:G-upperbound}, it follows immediately that $\overline{G}$ is also regularly varying at infinity with index $-1/\gamma$, and that thus $G\in\textsc{mda}(\gamma)$.

To compare $m'_n$ to $m_n$, note that if $k\leq m_{n/(1+\epsilon)}/2$, then since for all $t<m_{n/(1+\epsilon)}$ we have $\overline{F}(t) > \frac{t}{n/(1+\epsilon)}$, Equation \eqref{eq:G-upperbound} gives us that $\overline{G}(k) > \frac{1}{2(1+\epsilon)} (1+\epsilon) \frac{2k}{n} =k/n$. It follows that $m'_n>k$ for all $k\leq m_{n/(1+\epsilon)}/2$, and thus $m'_n\geq m_{n/(1+\epsilon)}/2$. By the regular variation property of $m_n$ (see Lemma \ref{lem:implict:func}), we have $m_{n/(1+\epsilon)} \sim (1+\epsilon)^{-\frac{\gamma}{\gamma+1}} m_n$. This means that for large enough $n$, we can pay an additional factor of $1+\epsilon$ to get $m_{n/(1+\epsilon)} > \frac{1}{1+\epsilon} (1+\epsilon)^{-\frac{\gamma}{\gamma+1}} m_n > \frac{1}{(1+\epsilon)^2} m_n$. We thus have, for large enough $n$:
$$
 m'_n \geq \frac{1}{2(1+\epsilon)^2} m_n.
$$

A bound in the other direction follows similarly.
% *************************
% This is an upper bound on $m'_n$, but is incorrect as is. We need to take care of the $1+\epsilon$ here too.
% *************************
% To compare $m'_n$ to $m_n$ from above, note that if $k\geq m_n/2$, since for all $t\geq m_n$ we have $\overline{F}(t)\leq t/n$, Equation \eqref{eq:G-lowerbound} gives us that $\overline{G}(k)\leq \frac{1}{2} \frac{2k}{n} = k/n$. By comparing to the definition of $m'_n$, it follows that $m'_n \leq k$ for all $k\geq m_n/2$, and thus $m'_n \leq m_n/2$.
\end{proof}

%%%%%%%%%%%%%%%%%%%%%%%%%%%%%%%%%%%%%%%%%%%%%%%%%%%%%%%%%%%%%%%%%%%%%%%%%%%%%%%%
\subsection*{\underline{Proof of Lemma \ref{lem:pointwise-bounds} }}
%%%%%%%%%%%%%%%%%%%%%%%%%%%%%%%%%%%%%%%%%%%%%%%%%%%%%%%%%%%%%%%%%%%%%%%%%%%%%%%%

\begin{proof}[\color{white}]
Before we proceed, we give a convenient representation of $\Sigma(\u)$ in an integral form. We have:
\begin{eqnarray} \label{eq:integral}
\Sigma(\u)
   &=& \frac{1}{\ln 2} \int_1^\infty \IND\{y>\u\} \ln(1+y-\u) \PROB(\mathd y) \nonumber \\
   &=& \frac{1}{\ln 2} \int_1^\infty \IND\{y>\u\} \int_1^\infty \frac{\IND\{\u<x<y\}}{1+x-\u} \mathd x \PROB(\mathd y) \nonumber  \\
   &=& \frac{1}{\ln 2} \int_1^\infty \IND\{x>\u\} \frac{1}{1+x-\u}  \int_1^\infty \IND\{y>x\} \PROB(\mathd y) \mathd x \nonumber \\
   &=& \frac{1}{\ln 2} \int_\u^\infty \frac{\overline{G}(x)}{1+x-\u} \mathd x,
\end{eqnarray}
where we have written an integral form of the logarithm and used Fubini's theorem to swap the integrals.

When $G$ belongs to an envelope class defined by $F$, we have $\overline{G}\leq\overline{F}$, and therefore we can see from Equation \eqref{eq:integral} that $\Sigma(\u)$ under $G$ is dominated by that under $F$. In particular, when $F\in \textsc{mda}(\gamma)$ with $\gamma>0$, it admits logarithmic moments, and we trivially see that $\Sigma(\u)$ is finite. But what we are really interested in is the decay of $\Sigma(\u)$ as $\u$ grows.

Equation \eqref{eq:integral} shows that the decay of $\Sigma(\u)$ is governed by the decay of $\overline{G}(\u)$ itself, which dominates for small values of $x$, and is then complemented by the decay of $1/(1+x-\u)$. We can capture this compromise by splitting the integral at some arbitrary point, say $\u+t-1$ for some $t\geq 1$. We have:
\begin{eqnarray}
\Sigma(\u)
    &=& \frac{1}{\ln 2} \int_\u^{\u+t-1} \frac{\overline{G}(x)}{1+x-\u} \mathd x \\
&& +  \frac{1}{\ln 2} \int_{\u+t-1}^\infty \frac{\overline{G}(x)}{1+x-\u} \mathd x \\
    &\leq& \overline{G}(\u) \log(t) + \frac{1}{\ln 2} \int_t^\infty \frac{\overline{G}(y)}{y} \mathd y, \label{eq:splitting}
\end{eqnarray}
where we have split the integral, bounded $\overline{G}$ in both parts by its largest value, and performed the integration of the first part and a change of variable in the second. For the latter, we proceed by first bounding by the envelope:
$$
\int_t^\infty \frac{\overline{G}(y)}{y} \mathd y \leq \int_t^\infty \frac{\overline{F}(y)}{y} \mathd y.
$$

We would now like to invoke Karamata's integration theorem. Let us make the change of variable $y=m(z)$ and set $z_0=m^{-1}(t)$. By using the property that $\overline{F}(m(z))=m(z)/z$ and by performing an integration by parts, we get:
$$
\int_t^\infty \frac{\overline{F}(y)}{y} \mathd y = \int_{z_0}^\infty \frac{\mathd m(z)}{z} = -\frac{m(z_0)}{z_0}+\int_{z_0}^\infty \frac{m(z)}{z^2} \mathd z.
$$
Now note that $\frac{m(z)}{z^2}$ is regularly varying with index $\frac{\gamma}{\gamma+1}-2$, by Lemma \ref{lem:implict:func}. By using Karamata's integration theorem, Theorem \ref{thm:karamata}, we find that given $\epsilon>0$, for large enough $t$,
$$
\int_{z_0}^\infty \frac{m(z)}{z^2} \mathd z \leq (\gamma+1+\epsilon\ln 2) \frac{m(z_0)}{z_0)}.
$$
Combining the last three equations together, we have:
$$
\int_t^\infty \frac{\overline{G}(y)}{y} \mathd y \leq (\gamma +\epsilon\ln 2) \frac{m(z_0)}{z_0},
$$
and the claim follows using the fact that $m(z_0)/z_0=\overline{F}(m(z_0))=\overline{F}(t)$.
\end{proof}

%%%%%%%%%%%%%%%%%%%%%%%%%%%%%%%%%%%%%%%%%%%%%%%%%%%%%%%%%%%%%%%%%%%%%%%%%%%%%%%%
\subsection*{\underline{Proof of Lemma \ref{lem:mean} }}
%%%%%%%%%%%%%%%%%%%%%%%%%%%%%%%%%%%%%%%%%%%%%%%%%%%%%%%%%%%%%%%%%%%%%%%%%%%%%%%%

\begin{proof}[\color{white}]
Recall that:
$$
  \EXP[\overline{G}(M_i)]=\PROB\{X_{i+1}>M_i\} = \EXP[\IND\{ X_{i+1}>M_i \}].
$$

We would like to exploit the independence structure (in fact, only the exchangeability aspect of independence). To make this symmetry complete for the event of interest, in what follows we effectively replace $M_i$ by a new threshold, equal to $M_{i+1}-1$.

Let $\varsigma$ indicate a uniform random permutation of $1,\cdots,i+1$ that we inject into the probability space. Note that $M_i$ never decreases and increases at most by $1$ at every new sample (see also the property of self-boundedness in the Appendix of the paper). Furthermore $M_{i+1}$ is permutation invariant, as its definition relies only on order statistics. In light of these properties, we can write:
\begin{eqnarray*}
  \PROB\{X_{i+1}>M_i\} &\leq &\PROB\{X_{i+1}>M_{i+1}-1\} \\ &=& \PROB\{X_{\varsigma(i+1)}>M_{i+1}-1\}.
\end{eqnarray*}

Let us now condition on the values of the samples $X_1,\cdots,X_{i+1}$. This fixes the value of $M_{i+1}-1$, by invariance. The only randomness that remains in the last expression is that due to permutations. Now note that the event $\{X_{\varsigma(i+1)}>M_{i+1}-1\}$ occurs a fraction of times corresponding to the number of samples strictly larger than $M_{i+1}-1$, or equivalently greater than or equal to $M_{i+1}$. Thus:
\begin{eqnarray*}
\lefteqn{\PROB\{X_{\varsigma(i+1)}>M_{i+1}-1| X_1,\cdots,X_{i+1}\}} \\
&= &\frac{ \sum_{j=1}^{i+1} \IND\{X_j>M_{i+1}-1\} }{i+1}.
\end{eqnarray*}

Finally, observe that we have:
\begin{eqnarray*}
\sum_{j=1}^{i+1} \IND\{X_j>M_{i+1}-1\}
  &=&    \max\{0, ~k:X_{k,i+1}\geq M_{i+1}\} \\
  &\leq& M_{i+1},
\end{eqnarray*}
where the inequality follows from the fact that all order statistics of rank greater than or equal to $M_{i+1}$ are no greater than $M_{i+1}$, by the definition $M_{i+1}=\min\{i+1, ~k:X_{k,i+1}\leq k\}$.

Therefore, as claimed:
\begin{eqnarray*}
\EXP[\overline{G}(M_i)] &\leq &\EXP\left[\PROB\{X_{\varsigma(i+1)}>M_{i+1}-1| X_1,\cdots,X_{i+1}\}\right]\\
 &\leq &\frac{\EXP[M_{i+1}]}{i+1}. \qedhere
\end{eqnarray*}
\end{proof}

%%%%%%%%%%%%%%%%%%%%%%%%%%%%%%%%%%%%%%%%%%%%%%%%%%%%%%%%%%%%%%%%%%%%%%%%%%%%%%%%
\subsection*{\underline{Proof of Lemma \ref{lem:elias-code} }}
%%%%%%%%%%%%%%%%%%%%%%%%%%%%%%%%%%%%%%%%%%%%%%%%%%%%%%%%%%%%%%%%%%%%%%%%%%%%%%%%

\begin{proof}[\color{white}]
We have that given $\epsilon$, then beyond some $i_0$:
\begin{eqnarray*}
\EXP[\ell(C_{\textsc{e}})]
    &\leq& 2\sum_{i=1}^{n-1} \left( \EXP[\Sigma(M_i)] + \rho\EXP[G(M_i)]\right) \\
    &\leq& 2\sum_{i=1}^{i_0-2} \left( \EXP[\Sigma(M_i)] + \rho\EXP[G(M_i)]\right) \\
    && \quad +\ (2+\epsilon')\sum_{i=i_0-1}^{n-1} \frac{m_{i+1}}{i+1}  \log(m_{i+1}) \\
&& + (\gamma/\ln 2+\rho+\epsilon') \frac{m_{i+1}}{i+1} \\
    &\leq& \kappa + (2+\epsilon) \sum_{i=i_0-1}^{n-1} \frac{m_{i+1}\log(m_{i+1})}{i+1},
\end{eqnarray*}
where we have combined Equations \eqref{eq:EGUi-mi} and \eqref{eq:ESUi-mi} into Equation \eqref{eq:total-bound}, and where the adjustment between $\epsilon$ and $\epsilon'$ is made a priori.

To establish the integral bound, we are at first tempted to assume that $(m_i \ln m_i)/i$ is non-increasing. However, this is not strictly true. The furthest $m_{i+1}$ will move away from $m_i$ is when $\overline{F}$ remains constant (equal to $\frac{m_i}{i}$) between $m_i$ and $m_{i+1}$. This would mean that $m_{i+1}/(i+1)=\frac{m_i}{i}$, and thus $m_{i+1}\leq(1+1/i)m_i$. From this, we find that:
\begin{eqnarray*}
  \frac{m_{i+1} \log m_{i+1}}{i+1} &\leq &\frac{m_i \log m_i}{i} + \frac{m_i \log (1+1/i)}{i} \\ &\leq& \frac{m_i \log m_i}{i} + \frac{m_i}{i}.
\end{eqnarray*}
Therefore, the integral can deviate from the sum by at most $\sum \frac{m_i}{i}$, which is of negligible order compared to the magnitude of the sum.

The direct bound follows by noting that we can use Jensen's inequality and the fact that $m(n)$ is non-decreasing, to show:
\begin{eqnarray*}
\lefteqn{  \int_1^n \frac{1}{t} m(t) \log m(t) \mathd t}\\
      &\leq& \left(\int_1^n \frac{1}{t} m(t) \mathd t\right) \log \left(  \frac{\int_1^n \frac{1}{t} m^2(t) \mathd t}{\int_1^n \frac{1}{t} m(t) \mathd t}    \right)\\
      &\leq&  m(n) \log(n) \log m(n).
\end{eqnarray*}

Lastly, to specialize to the Fr\'echet case, recall (by Lemma \ref{lem:implict:func}) that $m(t)$ is $\textsc{rv}_{\gamma/(\gamma+1)}$, therefore we also have that $(m(t) \log m(t)) /t$ is $\textsc{rv}_{-1/(\gamma+1)}$. Karamata's integration theorem, Theorem \ref{thm:karamata}, then tells us that given $\epsilon>0$, there exists a $t_0$ and $t_1>t_0$ such that for all $n>t_1$:
\begin{equation} \label{eq:karamata-frechet}
\frac{m_n \log m_n}{\int_{t_0}^n \frac{1}{t} m(t) \log m(t) \mathd t} \geq (1-\epsilon) \left[\frac{-1}{1+\gamma}+1\right] = (1-\epsilon) \frac{\gamma}{1+\gamma}.
\end{equation}

When $\gamma>0$, we can therefore combine Equations \eqref{eq:elias-bound-integral} and \eqref{eq:karamata-frechet} to write that there exists a constant $\kappa$ such that for large enough $n$:
\begin{eqnarray*}
\EXP[\ell(C_{\textsc{e}})]
  &\leq& \kappa + (2+\epsilon)\frac{\gamma+1}{\gamma} m_n \log m_n \\
  &\leq& (2+o_{\Lambda}(1)) \frac{\gamma+1}{\gamma}  m_n \log m_n \\
  &\leq& (2+o_{\Lambda}(1)) m_n \log n,
\end{eqnarray*}
where for the last expression we have used the regular variation property of $m_n \sim n^{\gamma/(\gamma+1)} L_m(n)$, for some slowly varying function $L_m$ (given in Lemma \ref{lem:implict:func}), and the fact that $\log L(n)/\log n \to 0$ for any slowly varying function $L$.
\end{proof}

\end{document}